%% file: main.tex
\documentclass{article}

\include{macros}
\usepackage[figuresright]{rotating}

\begin{document}

\title{A search-to-decision reduction for the linear code equivalence problem}

\author{Jean-Fran\c{c}ois Biasse
 \orcidlink{0000-0001-8591-8408},
  \and
  Giacomo Micheli
 \orcidlink{0000-0002-5265-5207},
  \and
  Benjamin Prada
 \orcidlink{0009-0003-5329-5973},
  \and
  Philip Waitkevich
\orcidlink{0009-0008-6481-2100}
\thanks{
Department of Mathematics and Statistics, University of South Florida, Tampa,
FL 33620 USA. E-mail: \{biasse,gmicheli,bprada,phillipwaitkevich\}@usf.edu.}
 \thanks{An extended abstract of this paper was presented at the 2023 IEEE ISIT conference.}
}

%\IEEEpubid{0000--0000~\copyright~2023 IEEE}

%% e.g. \dochead{17th International Conference on Dynamical Processes in Excited States of Solids}

%\title{}


\date{}

\maketitle

\begin{abstract}
We present a polynomial-time reduction from the search variant of the linear code equivalence problem (i.e. the search for a linear isometry between the inputs) to its decisional variant. More precisely, given two linearly equivalent codes $\mathcal C_1,\mathcal C_2 \subseteq \mathbb{F}_q^n$, we show how to recover a linear isometry between them by making a polynomial number of queries to an oracle for decisional linear code equivalence.
\begin{comment}
In this paper, we present an algorithm to find a linear isometry between two input linear codes $\mathcal C_1,\mathcal C_2$ over $\F_q^n$ by making a polynomial number of queries to the decisional linear code equivalence problem (i.e. the problem of deciding if there is a linear isometry between two input codes). 
\end{comment}
First, we prove that search-Permutation Code Equivalence (search-PCE -- the problem of finding a permutation $\pi\in\mathcal S_n$ mapping $\mathcal C_1$ to $\mathcal C_2$) reduces in polynomial time to PCE (i.e. the problem of deciding if there is a permutation map from $\mathcal C_1$ to $\mathcal C_2$) via at most $n^2$ oracle calls on instances of dimension $k$ and length at most $n^2(n+1)/2$. We then extend this approach to linearly equivalent codes: we recover the permutation part of a linear isometry via at most $n^2$ calls to a Linear Code Equivalence (LCE) oracle on instances of the same size, and we give a deterministic polynomial-time algorithm to recover the diagonal part once this permutation is known. Altogether, this yields a polynomial-time procedure to recover a linear isometry from an oracle for decisional LCE. From a linear-algebraic perspective, our results provide an explicit reconstruction
of a monomial equivalence between two matrix representations from oracle access
to the corresponding orbit membership problem.
\begin{comment}
More specifically, we prove two statements. We show that there is a polynomial algorithm that finds a permutation between two permutationally equivalent codes of dimension $k$ over $\F_q^{n}$ by using at most $n^2$ calls to an oracle that decides if two codes of dimension $k$ and length at most $n^2(n+1)/2$ over $\F_q$ are permutationally equivalent. We then generalize this result to the case of linearly equivalent codes: We prove that there is a polynomial algorithm that finds a linear isometry between two linearly equivalent codes of dimension $k$ over $\F_q^{n}$ by using at most $n^2$ calls to an oracle that decides if two codes of dimension $k$ and length at most $n^2(n+1)/2$ over $\F_q$ are linearly equivalent.
\end{comment}
\end{abstract}

% \linenumbers

%\section{}
%\label{}

\input{Introduction}
\input{background}

\input{PCE}

\input{Permutation}
\input{Multipliers}



%% References


\bibliographystyle{elsarticle-num}
\bibliography{ResearchOutline}






% \end{thebibliography}

\end{document}

%% file: macros.tex
\usepackage[utf8]{inputenc} % set input encoding (not needed with XeLaTeX)

\usepackage{geometry} % to  change the page dimensions
\usepackage{graphicx} % support the \includegraphics command and options

\usepackage{indentfirst} % Activate to add indent to first line after each section
\usepackage{booktabs} % for much better looking tables
\usepackage{array} % for better arrays (eg matrices) in maths
\usepackage{paralist} % very flexible & customisable lists (eg. enumerate/itemize, etc.)
\usepackage{verbatim} % adds environment for commenting out blocks of text & for better verbatim
\usepackage{subfig} % make it possible to include more than one captioned figure/table in a single float
\usepackage{multicol} % allow using multiple columns per page
\usepackage{amsthm}
\usepackage{amssymb}
\usepackage{amsmath}
\usepackage{mathtools}
\usepackage{color, colortbl}
\usepackage{url}
\usepackage{orcidlink}
\usepackage[boxed, linesnumbered, noend]{algorithm2e}
\usepackage{algorithmic}

\usepackage{fancyhdr} % This should be set AFTER setting up the page geometry
\usepackage{sectsty}
\allsectionsfont{\sffamily\mdseries\upshape} % (See the fntguide.pdf for font help)
\usepackage[nottoc,notlof,notlot]{tocbibind} % Put the bibliography in the ToC
\usepackage[titles,subfigure]{tocloft} % Alter the style of the Table of Contents

\newcounter{thmcount}
\newcounter{defcount}
\newcounter{remcount}
\newcounter{factcount}
\theoremstyle{plain}
\newtheorem{thm}[thmcount]{Theorem}
\newtheorem{lemma}[thmcount]{Lemma}
\newtheorem{corollary}[thmcount]{Corollary}
\newtheorem{proposition}[thmcount]{Proposition}
\newtheorem{fact}[factcount]{Fact}

\newtheorem{definition}[defcount]{Definition}
\newtheorem{rmk}[remcount]{Remark}

\DeclarePairedDelimiter{\gen}{\langle}{\rangle}

\newcommand{\F}{\mathbb{F}}
\newcommand{\C}{\mathcal{C}}
\newcommand{\Sn}{\mathcal{S}}

\newcommand{\permequiv}{\stackrel{P}{\approx}}
\newcommand{\diag}{\text{diag}}

%% file: Introduction.tex
\section{Introduction}

Given two linear codes $\mathcal{C}_1$ and $\mathcal{C}_2$ of length $n$  over 
a finite field $\F_q$, the \textit{Permutation Code Equivalence} problem (PCE) consists in finding a permutation map $\pi\in\mathcal{S}_n$ such that $\pi(\mathcal{C}_1) = \mathcal{C}_2$. If we are given the image of a basis of $\mathcal{C}_1$ via $\pi$, this problem is computationally easy. On the other hand, if we only know an arbitrary basis of $\pi(\mathcal{C}_1)$, then there is no known efficient algorithm for solving general instances of PCE. In most of the literature, PCE is phrased in its decisional variant: given $\mathcal{C}_1$ and $\mathcal{C}_2$, decide whether there is $\pi\in\mathcal{S}_n$ such that $\mathcal{C}_2 = \pi(\mathcal{C}_1)$. On the other hand, the search variant of PCE requires us to find $\pi$ such that $\pi(\mathcal{C}_1) = \mathcal{C}_2$. In the following, we call these two variants PCE and search-PCE. PCE can be generalized to general linear isometries between codes. The Linear Code Equivalence Problem (LCE) is the task of deciding if there is a linear isometry $\tau\in \F_q^{*n}\rtimes \mathcal{S}_n$ such that $\tau(\mathcal{C}_1) = \mathcal{C}_2$, while the most generic variant of the code equivalence problem consists in deciding if there is a field isomorphism $\alpha\in \operatorname{Aut}(\F_q)$ and $\tau\in \F_q^{*n}\rtimes \mathcal{S}_n$ such that 
$\tau(\alpha(\mathcal{C}_1)) = \mathcal{C}_2$. In the rest of this paper, we focus on PCE and LCE. Interestingly, the available algorithms for the resolution of the most generic instances of the code equivalence problem actually do it by solving the search variant. From a linear-algebraic point of view, the problem is to reconstruct a monomial transformation between two matrix representations of equivalent codes; our arguments exploit repeated or proportional column structure and, in the linear case, the behavior of supports under row reduction. In particular, the code equivalence problem can be viewed as a matrix equivalence
problem under the action of the monomial group, combining left multiplication by
$\mathrm{GL}_k(\mathbb{F}_q)$ and right multiplication by permutation and diagonal
matrices. Our reductions exploit structural properties of this action, such as
the behavior of repeated or proportional columns and the invariance of row-reduced
forms, to reconstruct a representative of the equivalence class. In this paper, we show that search-PCE can be efficiently solved by using an oracle to PCE, and that search-LCE can be efficiently solved by using an oracle to LCE.

\paragraph{\textbf{The Code Equivalence problem in cryptography}} The code equivalence problem intersects cryptography based on error correcting codes (a.k.a ``code-based cryptography''). In the McEliece cryptosystem~\cite{mceliece1978public}, the private key is a code $\mathcal{C}_1$, and the public key is $\mathcal{C}_2 = \pi(\mathcal{C}_1)$ where $\pi$ is a secret permutation. Note that being able to solve PCE (or search-PCE) does not imply being able to break the McEliece cryptosystem because in this case, we are only given $\mathcal{C}_2$. If potential candidates for $\mathcal{C}_1$ are given to us (for example in the context of a key exposure attack), then we can decide if any of them is the private key by solving PCE, but in general there are too many possible secret keys to efficiently reduce the cryptanalysis of McEliece to PCE (or search-PCE). 

A Zero-Knowledge protocol due to Girault~\cite{Girault} first used 
PCE as a hardness assumption. This identification scheme both relies on 
the hardness of the binary syndrome decoding problem and of search-PCE. This scheme was later improved by Sendrier and Simos~\cite[Sec. 4.2]{Sendrier-Simos}. Biasse, Micheli, Persichetti and Santini described the  identification (and digital signature) scheme LESS~\cite{LESS} which uses search-PCE (or search-LCE) as 
its sole security assumption (through a rigorous security proof). From a high level 
perspective, the task of solving search-PCE (resp. search-LCE) (and thus breaking LESS) can be viewed as inverting the group action of $\mathcal{S}_n$ (resp. $\F_q^{*n}\rtimes \mathcal{S}_n$) on the set of linear codes equivalent to a given code $\mathcal{C}_1$. From this perspective, cryptography based on the hardness of search-PCE and search-LCE falls under the umbrella of ``hard homogeneous spaces'' described by Couveignes~\cite{Couveignes06}, and ``Cryptographic Group Actions'' (see~\cite{AlamatiFMP20}). Hence, many results applying to this family of schemes, including follow-up works from the isogeny-based scheme CSIDH~\cite{CSIDH}, can be used to enhance code equivalence-based cryptosystems. This way, practical improvements to LESS were described in~\cite{LESS-FM}, as well as advanced signature functionalities (ring signatures, identity based signatures) enabled by the cryptographic group action framework~\cite{LESS-advanced}.

\paragraph{\textbf{Prior work on the resolution of the code equivalence problem}} In 1982, Leon~\cite{Leon1982} described an algorithm to solve search-PCE that relies on the following observation: since $\pi$ preserves the Hamming weight of a codeword, we can attempt to draw pairs $c,\pi(c)$ by computing the sets of weight-$w$ codewords in $\mathcal{C}_1$ and $\mathcal{C}_2$ and matching them when they have the same multisets of entries. With enough valid pairs, $\pi$ can be recovered. The complexity of Leon's algorithm is exponential, but this approach seems to offer the best performance in practice in the general case. Since the publication of cryptosystems based on the code equivalence problem, some follow-up works brought practical improvements to Leon's method~\cite{CEcrypto,Beullens20} to solve search-PCE and search-LCE without changing the overall complexity class. 

Besides Leon's algorithm and related works, algorithms exist to solve search-PCE efficiently for codes of small hull (the intersection of a code with its dual). For example, Sendrier's Support Splitting Algorithm (SSA~\cite{Sendrier-SSA})  relies on the concept of a signature function that takes a code as input, and returns the same value for 
two equivalent codes. Sendrier's signature function is based on the \textit{weight enumerator function} whose calculation takes time $\tilde{O}\left( n^3 + q^h n^2\right)$ where $h$ is the dimension of the hull. In the case of a hull of small dimension, this algorithm can be efficient, however, if a code is weakly self-dual, then the complexity is exponential. Note that SSA can be extended to the resolution of search-LCE by considering 
the \textit{closure} of the input codes and applying SSA. However, when $q\geq 5$, the closure of a code is always 
weakly self-dual, which means that for the vast majority of instances of search-LCE, SSA has 
exponential complexity. The 
work of Bardet, Otmani, and Saheed-Taha~\cite{BardetOS19} shows that when $h=0$, an instance of search-PCE efficiently reduces to an 
instance of the Weighted Graph Isomorphism (WGI) problem. The resulting complexity is in $O\left( n^{2.373}\text{C}_{\text{WGI}}(n)\right)$, where $\text{C}_{\text{WGI}}(n)$ denotes the cost of solving WGI in a graph of size $n$, which can be done in quasi-polynomial time with Babai's algorithm~\cite{Babai16}. An algebraic approach was 
also investigated in~\cite{Saheed-Thesis} where the code equivalence problem was reduced to the resolution of a system of quadratic equations. With trivial hull, the system can be solved efficiently via the \textit{block linearization} technique, but in the case of an arbitrary hull, the complexity appears to be exponential. The deterministic algorithm with the best asymptotic run time to solve PCE is due to Babai~\cite{babai2011code}, and it runs in time $2^{n + o(n+q)}$. This work was generalized by Bennett et al.~\cite{ISIT25_CE} who presented a $2^{n+o(n+q)}$-time deterministic algorithm for search-LCE, as well as randomized $2^{n/2 + o(n+q)}$-time algorithms for search-PCE and search-LCE, and $2^{n/3 + o(n+q)}$-time quantum algorithms for search-PCE and search-LCE.  

In terms of complexity theory, it is worth noting that code equivalence is unlikely to be $\operatorname{NP}$-complete. More specifically, Petrank and Roth showed in~\cite{PetrankR97} that if code equivalence were $\operatorname{NP}$-complete, then the polynomial hierarchy would collapse at the first level, which is widely believed to be untrue. In the same paper, they showed that the Graph Isomorphism (GI) problem efficiently reduced to PCE. Until Babai's groundbreaking work on GI~\cite{Babai16}, this was often used to justify the alleged hardness of PCE. However, this is no longer relevant to gauge the computational hardness of PCE (except in the very special case of $h=0$, an efficient solution to GI does not seem to imply an efficient solution to PCE). 

Finally, it is worth noting that generalizations of the code equivalence problem to rank metric codes have recently received some attention. 
Recent work proposed optimized methods for the resolution of the code equivalence problem in rank metric codes~\cite{Couvreur21,Reijnders22}, 
and Chou et al. suggested the use of this problem as the hardness assumption for a new identification scheme~\cite{MEDS}.

%Mention complexity results, connection to GI. 

\paragraph{\textbf{Our contribution}}
\begin{comment}In this paper, we prove the following reductions from search PCE to PCE and search LCE to LCE respectively.
\begin{thm}
 There is a polynomial time algorithm to solve search-PCE on input $[n,k,d]_q$ codes that uses at most $n^2$ calls to an oracle to solve PCE on input codes of dimension $k$ and length at most $n^2(n+1)/2$ over $\F_q$.
\end{thm}

\begin{thm}\label{th:PCE}
 There is a polynomial time algorithm to solve search-LCE on input $[n,k,d]_q$ codes that uses at most $n^2$ calls to an oracle to solve LCE on input codes of dimension $k$ and length at most $n^2(n+1)/2$ over $\F_q$.
\end{thm}
The above statements mean that a given instance $\mathcal{C}_1,\mathcal{C}_2$ of search-PCE (resp. search LCE) can be 
solved by a polynomial number of calls to an oracle that solves PCE (resp. LCE). Search-to-decision reductions are known 
for many computational problems, including the Graph Isomorphism problem~\cite[Chap. 1]{KST93}, which is related to 
PCE. However, despite extensive prior work on the resolution of the code equivalence problem, no reduction from search-PCE to PCE or from search -LCE to LCE was 
known. 
\end{comment}

Despite these advances, all existing algorithms for search-PCE and search-LCE operate independently of any oracle for their decisional counterparts. A natural question -- standard in complexity theory but largely unaddressed for code equivalence -- is whether a solver for the decisional problem can be leveraged to solve the search problem. We answer this affirmatively.

First, we prove that search-PCE reduces to PCE with a polynomial number of oracle calls.

\begin{thm}\label{th:main-PCE}
There is a polynomial-time algorithm to solve search-PCE on input $[n,k,d]_q$ codes that uses at most $n^2$ calls to an oracle for PCE on input codes of dimension $k$ and length at most $n^2(n+1)/2$ over $\mathbb F_q$.
\end{thm}

Second, we prove that search-LCE reduces to LCE in the sense that one can recover the permutation part of a linear isometry with a polynomial number of oracle calls.

\begin{thm}\label{th:main-LCE}
There is a polynomial-time algorithm which, on input two linearly equivalent $[n,k,d]_q$ codes, recovers a permutation $\pi \in S_n$ such that there exists $v \in \mathbb F_q^{\ast n}$ with $\tau=(\pi,v)$ satisfying $\tau(C_1)=C_2$, using at most $n^2$ calls to an oracle for LCE on input codes of dimension $k$ and length at most $n^2(n+1)/2$ over $\mathbb F_q$.
\end{thm}

Finally, given such a permutation $\pi$, we prove that the diagonal part $v$ can be recovered deterministically in polynomial time. This yields a polynomial-time procedure to recover a full linear isometry from an oracle for LCE. The deterministic reconstruction of the diagonal
part of a monomial transformation from linear-algebraic data is a key technical contribution of this work. While a standard
approach based on solving linear systems is known, it is inherently heuristic
in general. We replace it with a method based on row-reduced echelon forms and
connectivity properties of column supports, yielding a fully deterministic
procedure.

\begin{comment}
This paper is organized as follows: Section~\ref{sec:background} introduces the necessary background and notations. Section~\ref{sec:PCE-red} proves Theorem~\ref{th:PCE}, i.e. the reduction from search-PCE to PCE. Then Section~\ref{sec:LEC-perm-search} shows, given two codes $\mathcal C_1,\mathcal C_2$ with $\mathcal C_2 = \tau(\mathcal C_1$ for some $\tau\in \F_q^{*n}\rtimes \mathcal{S}_n$, how to find a permutation $\pi\in\mathcal S_n$ such that there is $\tau' = (v,\pi)\in \F_q^{*n}\rtimes \mathcal{S}_n$ with $\mathcal{C}_2 = \tau'(\mathcal C_1)$. Finally, Section~\ref{sec:diagonal-entries} shows how to recover $v$ given $\pi$. The results from Section~\ref{sec:PCE-red} were presented at the 2023 IEEE International Symposium on Information Theory (ISIT 23), and appeared in a 5p extended abstract published in the ISIT 23 proceedings. 
\end{comment}

This paper is organized as follows. Section~\ref{sec:background} introduces the necessary background and notation. Section~\ref{sec:PCE-red} proves Theorem~\ref{th:main-PCE}, namely the reduction from search-PCE to PCE. Section~\ref{sec:LEC-perm-search} proves the corresponding reduction for LCE at the level of the permutation part: given two linearly equivalent codes, it shows how to recover a permutation $\pi \in S_n$ for which there exists $v \in \mathbb F_q^{\ast n}$ such that $(\pi,v)$ maps one code to the other. Section~\ref{sec:diagonal-entries} then addresses the recovery of the diagonal part. After recalling a folklore heuristic approach and a simpler projective-frame-based special case for context, we present a new deterministic polynomial-time method to recover $v$ from $\pi$, thereby completing the reconstruction of a full linear isometry.

\paragraph{\textbf{Relation to prior work.}}
A preliminary version of this work appeared as an extended abstract in the
proceedings of ISIT~2023~\cite{Reduction}. That version establishes the
search-to-decision reduction for permutation code equivalence, corresponding to
the results of Section~3 of the present paper. The current manuscript extends
these results in several directions. In particular, we develop a reduction for
linear code equivalence, which requires additional techniques to recover the
diagonal part of the transformation once the permutation is known. This leads
to the deterministic reconstruction method of Section~5, which has no analogue
in the conference version.

%% file: background.tex
\section{Background}\label{sec:background}

\begin{comment}
We set new terms to easily distinguish between types of code equivalence problems.

\begin{definition}[PCE]
Permutation Code Equivalence (PCE) is the problem of deciding whether two linear codes are permutationally equivalent given their generator matrices.
\end{definition}

\begin{definition}[LCE]
Linear Code Equivalence (LCE) is the problem of deciding whether two linear codes are linearly equivalent given their generator matrices.
\end{definition}

\begin{definition}[search-PCE]
Search-PCE is the extension of PCE which, upon permutation equivalence, requires returning the permutation that realizes the equivalence between the two linear codes.
\end{definition}

\begin{definition}[search-LCE]
Search-LCE is the extension of LCE which, upon linear equivalence, requires returning the linear isometry that realizes the equivalence between the two linear codes.
\end{definition}
\end{comment}

In this section, we recall essential facts about linear codes and the code equivalence problem, and we specify relevant notations. An $[n,k,d]_q$ code $\mathcal{C}$ is a $k$-dimensional $\F_q$-vector space in $\F_q^n$ such that
the smallest Hamming weight of a non-zero $c\in\mathcal{C}$ is $d$. A matrix $G$ of $\F_q^{k\times n}$ whose rows span $\mathcal{C}$ is called a \textit{generator matrix} for $\mathcal{C}$. Let us denote by $\operatorname{GL}_k(\F_q)$ the invertible matrices of dimension $k$ over $\F_q$. Given a generator matrix
$G\in\F_q^{k\times n}$ of $\mathcal{C}$, and $S\in\operatorname{GL}_k(\F_q)$, the matrix $SG$ is also a generator matrix for $\mathcal{C}$. We denote by $G_i$ the $i$-th column of $G$. Additionally, we denote a sequence of matrices with indices in superscripts: $A^{(1)},A^{(2)},\ldots,A^{(\ell)}$. If $D$ is a diagonal matrix, $D(i)$ denotes the coefficient $D_{i,i}$.

\begin{definition}[Permutation matrix]
Let $\pi\in\mathcal{S}$. The permutation matrix $P\in\F_q^{n\times n}$ satisfying  $P_{j,i}=1$ if $\pi(i)=j$ and $P_{i,j}=0$ otherwise acts via a right multiplication by permuting columns according to $\pi$. More specifically, if $A,B\in\F_q^{k\times n}$, and $B = AP$, then the column of index $\pi(i)$ of $B$ is the column of index $i$ of $A$.
\end{definition}

 Given the natural correspondence between 
permutations and their associated permutation matrices, we use the notation $vP$ to denote the vector-matrix multiplication between the row vector $v$ and the permutation matrix $P$, as well as $P(i)$ to denote the image of $i\in\{1.\ldots,n\}$ by the permutation associated with $P$. When $P\in\{0,1\}^{n\times n}$ is a permutation matrix, we say that $P\in\Sn_n$. We also define the \textit{support of a permutation} as the set of indexes where it does not act like the identity. With the above definition, two linear codes $\C_1,\C_2$ with generator matrices $G_1,G_2$ are permutationally equivalent if and only if there exist an invertible matrix $S$ and a permutation matrix $P$ such that $G_2 = SG_1P$. 

\begin{definition}[PCE]
Let $\C_1,\C_2$ be $[n,k,d]_q$ codes. The Permutation Code Equivalence problem (PCE) is the task of deciding if there is  $\pi\in\mathcal{S}_n$ such that $\C_2 = \pi(\C_1)$. Equivalently, if $G_1,G_2$ are generator matrices for $\C_1$ and $\C_2$, PCE is the task of deciding if there is an invertible matrix $S$ and a permutation matrix $P$ such that $G_2 = SG_1P$.
\end{definition}

In many references, the permutation code equivalence problem is formulated in its decisional variant (like above). The search variant (denoted search-PCE) consists in finding the permutation $\pi$ such that $\C_2 = \pi(\C_1)$, or equivalently finding $S,P$ such that $G_2 = SG_1P$. When $G_1$ and $G_2$ are generator matrices of permutationally equivalent codes, we denote this by $G_1\permequiv G_2$. In the rest of the paper, we often say that such a $P$ \textit{realizes the equivalence} between $G_1$ and $G_2$. Note that such a $P$ is not in general unique. When the automorphism groups of the codes $\C_1$ and $\C_2$ are not trivial (i.e. there exist $P\neq I$ and $S\in\operatorname{GL}_k(\F_q)$ such that $SG_1P=G_1$), the $P$'s realizing the equivalence between $G_1$ and $G_2$ are a coset of the automorphism group of $\C_1$. 
Permutation equivalence can be extended to codes that are image of each other by a linear isometry. Two such codes $\C_1,\C_2$ are said to be linearly equivalent, and this property is denoted by $\C_1 \stackrel{L}{\approx} \C_2$. 

\begin{proposition}
All linear isometries of $\F_q^n$ can be uniquely identified by $(\sigma,v)\in \mathcal{S}_n\rtimes \F_q^{*n}$. The image $y\in\F_q^n$ of $x\in\F_q^n$ is given by 
$$
y_{\sigma(i)} = v_{\sigma(i)}x_i.
$$
\end{proposition}

Hence, an isometry necessarily acts via a permutation of the columns, followed by the scaling of the entries by non-zero scalars. From a matrix point of view, a linear isometry $\tau = (\sigma,v)\in \mathcal{S}_n\rtimes \F_q^{*n}$ acts on the columns of $A\in\F_q^{k\times n}$ via right multiplication by the permutation matrix $P$ corresponding to $\sigma$, followed by the right multiplication by the diagonal matrix $D = \diag(v_1,\ldots,v_n)$: $APD$. A matrix of the form $PD$ is a \emph{monomial matrix}.

\begin{definition}[LCE]
Let $\C_1,\C_2$ be $[n,k,d]_q$-linear codes. The Linear Code Equivalence problem (LCE) is the task of deciding if there exists $\tau\in \mathcal{S}_n\rtimes \F_q^{*n}$ such that $\C_2 = \tau(\C_1)$.  Equivalently, if $G_1,G_2$ are generator matrices for $\C_1$ and $\C_2$, LCE is the task of deciding if there is an invertible matrix $S$, a permutation matrix $P$, and a diagonal matrix $D$ such that $G_2 = SG_1PD$. 
\end{definition}

As for permutation code equivalence, LCE is stated as the decisional variant of the linear code equivalence, while we denote by search-LCE the problem of finding $\tau\in \mathcal{S}_n\rtimes \F_q^{*n}$ such that $\C_2 = \tau(\C_1)$, or equivalently, finding $S,P,D$ such that $G_2 = SG_1PD$.

%\paragraph{\textbf{Reduced row echelon form.}}
\begin{definition}[Reduced row echelon form (RREF)]
Let $A\in \mathbb{F}_q^{m\times n}$. We say that $A$ is in \emph{reduced row echelon form} (RREF) if the following conditions hold:
\begin{enumerate}
    \item every nonzero row of $A$ has a leftmost nonzero entry equal to $1$;
    \item the leftmost nonzero entry of each nonzero row lies strictly to the right of the leftmost nonzero entry of the preceding nonzero row;
    \item each column containing a leftmost nonzero entry of a row has zeros in all its other positions;
    \item all zero rows, if any, occur below the nonzero rows.
\end{enumerate}
The leftmost nonzero entry of a nonzero row is called a \emph{pivot}, and its column is called a \emph{pivot column}. Columns that are not pivot columns are called \emph{non-pivot columns}. For a matrix $A$, we denote by $RREF(A)$ the unique reduced row echelon form row-equivalent to $A$.
\end{definition}

\begin{proposition}[Uniqueness of the RREF with prescribed pivots]
Let $A\in \mathbb{F}_q^{m\times n}$, and let $p_1<\cdots<p_r$ be column indices. There is at most one matrix $B$ such that:
\begin{enumerate}
    \item $B$ is row-equivalent to $A$;
    \item $B$ is in reduced row echelon form;
    \item the pivot columns of $B$ are exactly $p_1,\dots,p_r$.
\end{enumerate}
In particular, $RREF(A)$ is unique.
\end{proposition}

\begin{proposition}[Polynomial-time computation of the RREF]
Let $A\in \mathbb{F}_q^{m\times n}$, and suppose that the pivot columns of $RREF(A)$ are known. Then $RREF(A)$ can be computed in polynomial time by Gaussian elimination. More generally, given a candidate set of pivot columns, one can decide in polynomial time whether it is the pivot set of a reduced row echelon form row-equivalent to $A$, and if so compute that form.
\end{proposition}

%% file: PCE.tex
\section{Search to decision reduction for PCE}\label{sec:PCE-red}

In this section, we present our efficient reduction from search-PCE to PCE. In a nutshell, it consists in iteratively searching for $P(t)$ 
for $t=1,2,\ldots,n$ where $P$ is a permutation matrix such that 
$SGP=F$ for the generator matrices $G,F$ spanning the two input codes. The main idea is to test whether $P(t) = i_t$ by extending the generator matrix $G$ with repetitions of its $t$-th column, and by extending the generator matrix $F$ with repetitions of its $i_t$-th column. The PCE oracle returns a positive answer if $P(t) = i_t$, thus allowing us to compute the image of $P$ at $t$. Further extensions of $G$ and calls to the PCE oracle allow us to find all images of the secret permutation. Note that we cannot simply compute the image of $P$ at all $t$ with independent calls to the oracle because there might exist multiple solutions to search-PCE that do not agree with each other at all indices. 

% Nethertheless, a direct application of these ideas does not necessarily yield the intended result. For example, if the idea above is applied independently on different indices, it is possible that there exists a permutation $P$ with $SGP=F$ for some $S$ with $P(t)=i_t$, and another permutation 
%$P'$ with $S'GP'=F$ for some $S'$ and $P'(s)=i_s$ such that $P(s)\neq i_s$. 

We start with a lemma that allows us to assume a certain choice for the images of a solution $P$ at indices where $G$ has repeated columns. The idea is the following: if the columns $G_t,G_{t_1},\ldots,G_{t_k}$ 
are all identical, so will be the columns $F_{P(t)},F_{P(t_1)},\ldots,
F_{P(t_k)}$. Hence, a permutation $P'$ equal to $P$ on inputs other than $t,t_1,\ldots,t_k$ will also satisfy $SGP' = F$. 
\begin{lemma}\label{lemma:aux}
Let $G$ and $F$ be $k\times n $ matrices such that there exists an invertible $k\times k$ matrix and a permutation matrix $Q$ such that $SGQ=F$. Let $\mathcal{L}=\{j,j^{(1)},\dots, j^{(t)}\}\subseteq \{1,\dots,n\}$ be a set of indexes corresponding to columns of $F$ equal to the  $j$-th column of $F$. Let $\mathcal{T}=Q^{-1}(\mathcal{L})$. Then, for every bijection $\sigma$ between $\mathcal{T}$ and $\mathcal{L}$, there exists a permutation matrix $Q'$ such that the permutation corresponding to $Q'$ restricted to $\mathcal{T}$ is equal to $\sigma$, $Q'$ restricted to $\{1,\ldots,n\}\setminus\mathcal{T}$ is equal to $Q$, and $SGQ'=F$.
\end{lemma}
\begin{proof}
Simply observe that any permutation $P$ of columns of $G$ indexed by $\mathcal{T}$ leads to an automorphism of $G$ because $GP=G$. 
Now choose a permutation matrix $P$ such that 
for all $i\in\mathcal{T}$, $P(i)= Q^{-1}( \sigma(i))$, 
and for $i\notin \mathcal{T}$, $P(i)=i$.  
%Now choose $P$ to be a permutation matrix such that $PQ=\sigma$ when restricted to $L$, and set $Q_1=PQ$. 
Such a $P$ has support $\mathcal{T}$. Choose $Q'$ as the composition of $Q$ and $P$. This means that $Q'(i) = Q(P(i))=\sigma(i)$ for $i\in\mathcal{T}$ and in terms of matrices, $Q' = PQ$. It follows directly that we have $SGQ'=SGPQ=F$.
\end{proof}

Thanks to Lemma~\ref{lemma:aux}, if we know that $SGP=F$, and 
if the columns $G_t,G_{t_1},\ldots,G_{t_\ell}$ 
are all identical, then there are solutions $P'$ 
satisfying $P'(t) = P(t_i)$ for some $i\leq \ell$ and $P'(j)=P(j)$ for 
$j\neq t,t_1,\ldots,t_\ell$.

Next, we show how to compute the image of one element of $\{1,\ldots,n\}$ via $P$ by using the oracle. Without loss of generality, we assume that we need to compute the image of $1$. This means that, given $i_1\in\{1,\ldots,n\}$, we want to decide whether there is a permutation matrix $P$ such that 
$SGP=F$ for some $S\in\operatorname{GL}_k(\F_q)$ and $P(1) = i_1$. Let $m(A)$ be the maximum number of times a column in a matrix $A$ is repeated. Clearly, $m(G)=m(F)=m$ since $SGP=F$. Let $G_1,\dots G_n$ be the columns of $G$ and $F_1,\dots,F_n$ be the columns of $F$.
Let us append to $G$ the $m$ columns equal to $G_1$, thus constructing the matrix 
\[G^{(1)}:=(G\mid\underbrace{G_1,\dots, G_1}_m)\] 
and let us append to $F$ the $m$ columns equal to $F_{i_1}$, thus constructing the matrix 
\[F^{(1)}:=(F\mid\underbrace{F_{i_1},\dots, F_{i_1}}_m).\]
Let us ask the oracle if the code generated by $G^{(1)}$ is permutationally equivalent to the code generated by $F^{(1)}$. 
\begin{proposition}
The answer from the PCE oracle on input $\{G^{(1)},F^{(1)}\}$ determines if there is a $P$ such that $P(1) = i_1$: 
\begin{enumerate}
 \item[(i)] If negative, then $P(1)\neq i_1$ for any permutation matrix $P$ realizing the equivalence between $G$ and $F$.
 \item[(ii)] If positive, there are $S\in\operatorname{GL}_k(\F_q)$ and a permutation matrix $P$ with $SGP=F$ and $P(1)= i_1$. 
\end{enumerate}
\end{proposition}
\begin{proof}
Item~(i): If the answer is negative, then $P(1)\neq i_1$ for any permutation matrix $P$ realizing the equivalence between $G$ and $F$, as we now explain. Suppose the contrary, i.e. the answer is negative and there exists a $P$ realizing the equivalence between $G$ and $F$ such that $P(1)=i_1$. Then, set $P'$ to be the permutation of $\Sn_{n+m}$ that is equal to $P$ on $\{1,\dots,n\}$ and such that $P'(n+s)=n+s$ for all $s\in \{1,\dots, m\}$. Then we have 
\[SG^{(1)} P'=(SGP|SG_1,\dots, SG_1)=(F|F_{i_1},\dots, F_{i_1})=F^{(1)}.\]
So there is at least one permutation map realizing the equivalence between $G^{(1)}$ and $F^{(1)}$, and therefore the answer could not have been negative.

Item~(ii): Suppose now that the answer is positive, then we have to prove that $P$ can be chosen such that $P(1)= i_1$. Let $\ell$ be the number of times the column $G_1$ is repeated in $G$.
Let $P'$ be a permutation realizing the equivalence of $G^{(1)}$ and $F^{(1)}$ and $S'$ an invertible matrix such that $S'G^{(1)} P'=F^{(1)}$. Notice that there is only one column of $G^{(1)}$ repeated $m+\ell$ times, and only one column of $F^{(1)}$ repeated $m+\ell$ times. Hence, it must be that $P'(1)\in \{i_1,i_1^{(1)},\dots, i_1^{(\ell-1)},n+1,\dots, n+m\}$, where $\{i_1^{(1)},\dots, i_1^{(\ell-1)}\}$ is the set indexes of columns of $F$ equal to the $i_1$-th column. Using Lemma~\ref{lemma:aux}, this implies that $P'$ can be chosen such that $P'(1)=i_1$ and $P'(j)=j$ for any $j\in \{n+1,\dots, n+m\}$. Moreover, we also have that $P'(j)\in \{1,\dots ,n\}\setminus \{i_1\}$ for all $j\in\{2,\dots, n\}$. These arguments show that $P'$ can be constructed to 
ensure that its restriction $P$ to $\{1,\ldots,n\}$ belongs to $\mathcal{S}_n$ (i.e. $P(\{1,\ldots,n\})=\{1,\ldots,n\}$) and satisfies 
$P(1)=i_1$.
It follows immediately an equivalence between $G$ and $F$ given by $P$ since 
\[S'G^{(1)} P' = (SGP\mid SG_1,\dots, SG_1) =(F\mid F_{i_1},\dots, F_{i_1})=F_1\]
so in particular $SGP=F$.
\end{proof}

Let us now show by induction on the number of indices for which we know an image how to do the general step.
We assume that we found the set $\mathcal A_{T-1}$ of  $T-1$ 
indices whose images by a permutation $P^{(T-1)}\in \Sn_n$ realizing the equivalence between $G$ and $F$ are known. This means that we know that there exist $S^{(T-1)}\in\operatorname{GL_k}(\F_q)$ and a permutation matrix $P^{(T-1)}$ such that $S^{(T-1)}GP^{(T-1)}=F$ and $P^{(T-1)}(t)=i_t$ for all $t\in \mathcal A_{T-1}$ and $|\mathcal A_{T-1}|=T-1$. 
The reader should notice that it is important to carry over the information of the previous steps as we need to be able to extend the ``same'' permutation step by step (i.e., when we move to the next element in the domain, we need to avoid finding an image of another permutation -- for example in the case in which the automorphism group of the code  acts transitively, thus allowing any possible image for the index $i$).
In the rest of this section, we show that by using the oracle, we can decide whether there exists a $P^{(T)}\in\Sn_n$ with $S^{(T)}GP^{(T)}=F$ for some $S^{(T)}\in\operatorname{GL}_k(\F_q)$ such that $P^{(T)}(t)=i_t$ for some $t\not\in \mathcal A_{T-1}$ and $P^{(T)}(t)=i_t$ for all $t\in \mathcal A_{T-1}$. 
Let again $m(G)=m(F)=m$ be the maximum number of times a column of $G$ (or $F$) is repeated. 
Let us define the sets $\mathcal G_{\mathcal A_{T-1}}=\{G_t:\: t\in \mathcal A_{T-1}\}=\{G_{t_1},\dots G_{t_\ell}\}$ and $\mathcal F_{\mathcal A_{T-1}}=\{F_{P^{(T-1)}(t_1)},\dots F_{P^{(T-1)}(t_\ell)}\}$. 
Since we counted columns without repetitions, we must have $\ell\leq |\mathcal A_{T-1}|= T-1$.
Append to $G$: $m$ columns equal to $G_{t_1}$, $2m$ columns equal to $G_{t_2}$, and eventually $\ell m$ columns equal to $G_{t_\ell}$, and denote the resulting matrix by $G^{(T-1)}$. 
Analogously, append to $F$: $m$ columns equal to $F_{P(t_1)}$, $2m$ columns equal to $F_{P(t_2)}$, and eventually $\ell m$ columns equal to $F_{P(t_\ell)}$, thus obtaining $F^{(T-1)}$. Let $N_{T-1} := n+\sum_{h=1}^{\ell}hm$ be the number of columns of $G^{(T-1)}$ and $F^{(T-1)}$.
\begin{fact}
The matrices $G^{(T-1)}$ and $F^{(T-1)}$ span equivalent codes.
\end{fact}
\begin{proof}
Since there is $S$ such that $SGP=F$, we know that $G^{T-1}$ and $F^{T-1}$ 
span equivalent codes. Indeed, for such a $S$, we have $SG_{t_i}=F_{P(t_i)}$, so we can extend $P$ to a permutation $P'^{(T-1)}\in\Sn_{N_{T-1}}$ defined by  
$P'^{(T-1)}(i)=P(i)$ for $i\leq n$, and $P'^{(T-1)}(i)=i$ for $i>n$ to get $SG^{(T-1)}P'^{(T-1)}=F^{(T-1)}$.
\end{proof}

\begin{fact}\label{fact:Pprime}
There is a permutation matrix $P'^{(T-1)}\in\Sn_{N_{T-1}}$ and an invertible $k\times k$ matrix $S'^{(T-1)}$ such that 
$$S'^{(T-1)} G^{(T-1)} P'^{(T-1)} = F^{(T-1)}$$ 
satisfying $P'^{(T-1)}(t_i)=P^{(T-1)}(t_i)$ for all $t_i\in \mathcal A_{T-1}$.
\end{fact}

\begin{proof}
Any permutation $P'\in \Sn_{N_{T-1}}$ such that there is $S'$ 
with $S'G^{(T-1)}P'=F^{(T-1)}$, has the property that the associated invertible matrix $S'$ satisfies $S'G_{t_i}=F_{P'(t_i)}$ for all $t_i\in \mathcal A_{T-1}$. This is because the number of times the column $G_{t_j}$  is repeated is unique in the matrix $G^{T-1}$ for all $j\in\{1,\dots, \ell\}$ for $\ell := |\mathcal G_{\mathcal A_{T-1}}|$ and the number of times the column $F_{P'(t_j)}$ is repeated is unique for the matrix $F^{(T-1)}$, and equal to the number of times $G_{t_j}$ is repeated in $G^{(T-1)}$. This means that $\forall i\leq \ell$, $t_i\in \left.P'\right. ^{-1}\left\{ \mathcal{S}^{F^{(T-1)}}_{F_{P'(t_i)}}\right\}$ where $\mathcal{S}^{F^{(T-1)}}_{F_{P'(t_i)}}$ is the set of indices 
$j$ such that the $j$-th column of $F^{(T-1)}$ is equal to $F_{P'(t_i)}$. 
Hence, by Lemma~\ref{lemma:aux} applied to $\mathcal{L}:=\mathcal{S}^{F^{(T-1)}}_{F_{P'(t_i)}}$, there is a permutation $P'^{(T-1)}$ satisfying  $P'^{(T-1)}(t_i)=P^{(T-1)}(t_i)$ for all $t_i\in \mathcal A_{T-1}$ and $S'^{(T-1)}G^{(T-1)}P'^{(T-1)}=F^{(T-1)}$ for an invertible matrix $S'^{(T-1)}$.
\end{proof}
%Now, the two newly created matrices $G'$ and $F'$ span codes that are equivalent and such that any $P'$ that realizes the equivalence can be chosen in such a way that $P'(t)=i_t$ for all $t\in A$ because of Lemma \ref{lemma:aux}.

Now, let us  take an element $t\not\in \mathcal A_{T-1}$. Consider the column $G_{t}$. If $G_t\in \mathcal G_{\mathcal A_{T-1}}$, then $G_t=G_{i}$ for some $i\in \mathcal A_{T-1}$. Next,  find a column of $F$ equal to $F_{P^{(P-1)}(t)}$ with an index $j\not\in P(\mathcal A_{T-1})$ (such $j$ must exist, or otherwise columns are repeated a different number of times, which is not possible for equivalent generator matrices). Using Lemma~\ref{lemma:aux}, 
there exists $P^{(T)}$ realizing the equivalence between $G$ and $F$ with $P^{(T)}(i) = P^{(T-1)}(i)$ for $i\in \mathcal A_{T-1}$ and 
 $P^{(T)}(t)=j$. Then, extend $\mathcal A_{T-1}$ as $\mathcal A_T\leftarrow \mathcal A_{T-1}\cup \{t\}$.

Let us now suppose that $G_t\not\in \mathcal G_{\mathcal A_{T-1}}$.
We append to $G^{(T-1)}$ a number $(\ell+1)m$ of columns equal to $G_{t}$, thus obtaining $G^{(T)}$ and  to $F^{(T-1)}$ 
let us append $(\ell+1)m$ columns equal to $F_{i_t}$, obtaining $F^{(T)}$. Now, we ask the oracle if the two codes generated by $G^{(T)}$ and $F^{(T)}$ are equivalent. 

\begin{proposition}
 The answer from the PCE oracle allows us to decide if there is a $P_T$ realizing the equivalence between $G$ and $F$ while agreeing with $P_{T-1}$ on $\mathcal A_{T-1}$ and such that $P_T(t) = i_t$: 
\begin{enumerate}
 \item[(i)] If the answer is negative, then $P^{(T)}(t)\neq i_t$ for any permutation matrix $P^{(T)}\in\Sn_n$ realizing the equivalence between $G$ and $F$ and agreeing with $P^{(T-1)}$ on all inputs from $\mathcal A_{T-1}$.
 \item[(ii)] If the answer is positive, there are $S^{(T)}\in\operatorname{GL}_k(\F_q)$ and a permutation matrix $P^{(T)}\in\Sn_n$ with $S^{(T)}GP^{(T)}=F$, $P^{(T)}(i)=P^{(T-1)}(i)$ for $i\in \mathcal A_{T-1}$ and $P^{(T)}(t)= i_t$. 
\end{enumerate}
\end{proposition}

\begin{proof}
Item~(i): If the answer is negative, then there cannot be a permutation matrix $P^{(T)}$ agreeing with $P^{(T-1)}$ on $\mathcal A_{T-1}$ with $P^{(T)}(t)=i_t$ as we now explain. Let $N_T:= n + \sum_{h=1}^{\ell + 1}hm$ be the total number of columns of $G^{(T)}$. If $P^{(T-1)}$ could be extended to some $P^{(T)}\in\Sn_n$ with $P^{(T)}(t)=i_t$ then we would have that
\[S^{(T)}G^{(T)}P'^{(T)}=F^{(T)}\]
for $S^{(T)}$ such that $S^{(T)}GP^{(T)}=F$ by constructing $P'^{(T)}$ as $P'^{(T)}(j):=P^{(T)}(j)$ for all $j\leq n$ and $P'^{(T)}(j):=j$ for all $j>n$.
%
%\[S''_TG''_TP''_T=F''_T\]
%for some invertible $S''_T$ by constructing $P''_T$ as $P''_T(j)=i_j$ for all $j\in A_{T-1}\cup \{t\}$ (these images were already fixed since the permutation $P_{T-1}$ has already been defined on $A_{T-1}$ by the inductive step), and as $P(j)=j$ for all $j>N-(\ell+1)m$. 
So the answer of the oracle could not have been negative.

Item~(ii): If the answer is positive, then we have to prove that the partially identified permutation $P^{(T-1)}$ we have previously generated can be extended by setting $P^{(T)}(t)=i_t$. Since the answer is positive, the codes generated by $G^{(T)}$ and $F^{(T)}$ are equivalent. 
Let $P'^{(T)}$ be a permutation that realizes the equivalence, i.e. 
$S'^{(T)}G^{(T)}P'^{(T)}=F^{(T)}$ for some $S'^{(T)}\in \operatorname{GL}_k(\F_q)$. We are now going to derive from the permutation $P'^{(T)}\in\Sn_{N_T}$ realizing the equivalence between $G^{(T)}$ and $F^{(T)}$ a permutation $P^{(T)}\in\Sn_n$ that realizes the equivalence between $G$ and $F$ (i.e. such that there exists $S^{(T)}\in \operatorname{GL}_k(\F_q)$ with $S^{(T)}GP^{(T)}=F$) and such that $P^{(T)}(t)=i_t$ (so that $P^{(T-1)}$ can be extended to $i_t$ on $t\not \in \mathcal A_{T-1}$). In other words, we are going to show that if one has constructed a set $\mathcal A_{T-1}\subseteq \{1,\dots, n\}$ of correct images for a permutation $P^{(T-1)}$ realizing the equivalence, then there is a permutation realizing the equivalence between $G^{(T)}$ and $F^{(T)}$ that satisfies the following properties: 
\begin{itemize}
\item it agrees with $P^{(T-1)}$ on $\mathcal A_{T-1}$,
\item it maps $t$ to $i_t$,
\item it maps $\{1,\dots, n\}$ to $\{1,\dots, n\}$.
\end{itemize}
The restriction of such permutation to $\{1,\ldots,n\}$ will provide an equivalence $P^{(T)}$ between $G$ and $F$ satisfying $P^{(T)}(j)=i_j$ for all $j\in \mathcal A_{T-1}\cup \{t\}$ and therefore further extending $P^{(T-1)}$ on the larger set $\mathcal A_T=\mathcal A_{T-1}\cup\{t\}$.

Since the column vectors $G_t$ and $F_{i_t}$ are repeated the same unique number of times in $G^{(T)}$ (resp. in $F^{(T)}$), we necessarily have $S'^{(T)}G_t=F_{i_t}$. We know that $P'^{(T)}\in\Sn_{N_T}$ maps indices of columns of $G^{(T)}$ equal to $G_t$ to indices of columns of $F^{(T)}$ equal to $F_{i_t}$. In addition, using Lemma~\ref{lemma:aux}, we know that there are permutations that map these indices in any arbitrary order. In particular, we can assume that 
\begin{itemize}
    \item $P'^{(T)}(t)=i_t$
    \item $P'^{(T)}(h) = h$ for $h\in\{N_T-(\ell+1)m+1,\ldots,N_T\}$ (i.e. $h$ is in the block of the last $(\ell+1)m$ indices).
    \item $P'^{(T)}(h) = \sigma(h)$ for any bijection 
    $$\sigma:\{\text{indices of columns of $G$ equal to $G_t$}\}\setminus \{t\}\rightarrow\{\text{indices of columns of $F$ equal to $F_{i_t}$}\}\setminus\{i_t\}.$$
\end{itemize}
Moreover, such permutation $P'^{(T)}$ maps $\{1,\ldots,N_T-(\ell+1)m\}$ to $\{1,\ldots,N_T-(\ell+1)m\}$. Hence, its restriction to $\{1,\ldots,N-(\ell+1)m\} = \{1,\ldots,N_{T-1}\}$ is a permutation 
$P'^{(T-1)}$ realizing an equivalence between $G^{(T-1)}$ and $F^{(T-1)}$. But we already observed in Fact~\ref{fact:Pprime} that such a $P'^{(T-1)}$ 
could be assumed to agree on $\mathcal A_{T-1}$ with $P^{(T-1)}$ realizing the equivalence between $G$ and $F$, and satisfying $P^{(T-1)}(j)=i_j$ for all $j\in A_{T-1}$. Therefore, by defining $P^{(T)}$ as the restriction of $P'^{(T)}$ to $\{1,\ldots,n\}$, we proved that if the answer of the oracle is positive, there exists a permutation $P^{(T)}$ of the columns of $G$ that realizes the equivalence with $F$ and that can be chosen with the property $P^{(T)}(j)=i_j$ for all $j\in \mathcal A_T$.
\end{proof}

To conclude the reduction from search-PCE to PCE, we simply iterate these steps until $\mathcal A_T=\{1,\dots, n\}$. This procedure is outlined in Algorithm~\ref{alg:search-to-decision-PCE}.  At the end, we have identified one of the potential permutation matrices $P$ realizing the equivalence between $G$ and $F$. The corollary below gives the specific count of oracle calls performed during the course of this reduction, together with an upper bound on the input size. 

\begin{algorithm}[ht]
\caption{Search to decision reduction for PCE.}
\begin{algorithmic}[1]\label{alg:search-to-decision-PCE}
  \REQUIRE Generator matrices $F,G$ of two $[n,k,d]_q$ codes such that there are $S\in GL_k(\F_q)$ and a permutation\\ matrix $P$ with $F = SGP$.
  \ENSURE A permutation matrix $P$ such that there is $S\in GL_k(\F_q)$ with $F = SGP$.
\STATE $m\leftarrow$ maximum number of columns in $G$ that are equal to each other. 
\STATE $P'\leftarrow 0^{n\times n}$, $\mathcal A\leftarrow \{\}$, $\mathcal B\leftarrow\{\}$, $\mathcal{F}\leftarrow\{\}$, $\mathcal{G}\leftarrow \{\}$.
\FOR{$i=1,\ldots,n$}
\IF{$G_i$ is not equal to an element of $\mathcal{G}$}
\STATE Let $\{G_{t_1},\ldots,G_{t_\ell}\}=\mathcal{G}$, and $\{F_{P(t_1)},\ldots,F_{P(t_\ell)}\}=\mathcal{F}$.
\STATE $G^{(i-1)}\leftarrow G$, $F^{(i-1)}\leftarrow F$.
\FOR{$h\leq \ell$}
\STATE Append $hm$ copies of $G_{t_h}$ at the end of $G^{(i-1)}$.
\STATE Append $h m$ copies of $F_{P(t_h)}$ at the end of $F^{(i-1)}$.
\ENDFOR
\FOR{$j\notin B$}
\STATE $G^{(i)}\leftarrow G^{(i-1)}$, $F^{(i)}\leftarrow F^{(i-1)}$.
\STATE Append $(\ell+1) m$ copies of $G_{i}$ at the end of $G^{(i)}$.
\STATE Append $(\ell+1) m$ copies of $F_{j}$ at the end of $F^{(i)}$.
\IF{$G^{(i)}\permequiv F^{(i)}$}
\STATE $\mathcal A\leftarrow \mathcal A\cup \{i\}$, $\mathcal B\leftarrow \mathcal B\cup \{j\}$, $\mathcal{G}\leftarrow \mathcal{G}\cup \{G_i\}$, $\mathcal{F}\leftarrow \mathcal{F}\cup \{F_j\}$, $P_{j,i}\leftarrow 1$.
\STATE \textbf{break}
\ELSE
\STATE Let $t_h$ such that $G_i$ is equal to $G_{t_h}\in\mathcal{G}$. 
\STATE Find $j\notin B$ such that $F_j$ is equal to $F_{P(t_h)}$.
\STATE  $\mathcal A\leftarrow \mathcal A\cup \{i\}$, $\mathcal B\leftarrow \mathcal B\cup \{j\}$, $P_{j,i}\leftarrow 1$.
\ENDIF
\ENDFOR
\ENDIF
\ENDFOR
\RETURN $P$
\end{algorithmic}
\end{algorithm}

\begin{corollary}
Algorithm~\ref{alg:search-to-decision-PCE} performs at most $n^2$ calls to an oracle to PCE with input codes of dimension $k$ and length at most $n^2(n+1)/2$ over $\F_q$ and eventually returns a permutation matrix $P$ such that there exists an invertible matrix $S$ with $SGP=F$. 
\end{corollary}

%% file: Permutation.tex
\section{Search to decision reduction for LCE}\label{sec:LEC-perm-search}

%\subsection{Search-LCE to LCE Reduction for Hamming-Metric Codes}

Let $\C_1,\C_2$ be two $[n,k,d]_q$ codes. In this section, we show how to reduce the problem of finding $\pi\in\mathcal{S}_n$ such that there exists $\tau = (\pi,v)\in  \mathcal{S}_n\rtimes \F_q^{*n}$ for some $v\in \F_q^{*n}$ with $\C_2 = \tau(\C_1)$ to LCE.

Now let $F,G$ be generator matrices of two $[n,k]_q$-linear codes that are linearly equivalent. They satisfy the relation $SGPD = F$ for some invertible matrix $S$, permutation matrix $P$, and diagonal matrix $D$ with non-zero diagonal entries. We define the equivalence relation $\sim$ on $v_1, v_2 \in \mathbb{F}_q^n$ by $v_1 \sim v_2 \iff \exists \lambda \in \mathbb{F}_q^* : v_1 = \lambda v_2$. By proportionality, we shall mean equivalence by $\sim$.

We begin by introducing a property analogous to Lemma \ref{lemma:aux}. Notice that if columns $G_{t_1},\ldots, G_{t_j}$ of $G$ are mutually proportional, the columns $F_{P(t)},F_{P(t_1)},\ldots, F_{P(t_k)}$ will be as well. Thus, a permutation $P'$ equal to $P$ on inputs other than $t,t_1,\ldots,t_k$ results in $SGP'$ being proportional to $G$. Unlike previously, these columns are not guaranteed to be identical. However, due to proportionality, further rescaling of the columns is enough to cause them to match. This means that there exists a additional diagonal matrix $D'$ identical to $D$ on all inputs other than $t_1, \ldots , t_j$ such that $SGP'D' = F$.

\begin{comment}
Through this property, we can construct a new solution by starting with any permutation matrix $P'$ that is identical to $P$ on all inputs other than $t_1, \ldots , t_j$. As a result, the column list $(SGP')_{P(t_1)},\ldots, (SGP')_{P(t_j)}$ is the same as $(SGP)_{P(t_1)},\ldots, (SGP)_{P(t_j)}$, although each $(SGP')_{P(t_i)}$ may not equal $(SGP)_{P(t_i)}$ in general. It is true, however, that $(SGP')_{P(t_1)},\ldots, (SGP')_{P(t_j)}$ are proportional to $(SGP)_{P(t_1)},\ldots, (SGP)_{P(t_j)}$ respectively. Therefore, further rescaling is enough to cause the columns to match, which we can ensure by carefully modifying the diagonal matrix that multiplies $SGP'$. That is, there then exists a diagonal matrix $D'$ identical to $D$ on all inputs other than $t_1, \ldots , t_j$ such that $SGP'D' = F$.
\end{comment}

\begin{lemma} \label{lem:newpermutation}
Let $F,G$ be generator matrices of two $[n,k]_q$-linear codes such that there exist a non-singular $S \in GL_k(\mathbb{F}_q)$, permutation matrix $P\in \{0,1\}^{n\times n}$, and diagonal matrix $D \in \mathbb{F}_q^{* n\times n}$ satisfying $SGPD = F$. Let $\mathcal{L} = \{t_1,\ldots, t_j\}$ be a set of indices of proportional columns of $F$. Let $\mathcal{T} = P^{-1}(\mathcal{L})$. Then for every bijection $\sigma: \mathcal{T} \to \mathcal{L}$, there exists a permutation matrix $P' \in \{0,1\}^{n \times n}$ such that $P'$ restricted to $\mathcal{T}$ is $\sigma$, $P'$ restricted to $\{1,\ldots,n\}\setminus\mathcal{T}$ is $P$, and $SGP'D' = F$ for some diagonal matrix $D'\in \mathbb{F}_q^{*n\times n}$.
\end{lemma}

\begin{proof}
Begin by fixing an arbitrary bijection $\sigma: \mathcal{T} \to \mathcal{L}$. Let the permutation matrix $\tilde{P} \in \{0,1\}^{n\times n}$ be such that $\tilde{P}(i) = P^{-1}(\sigma(i))$ for all $i \in \mathcal{T}$ and $\tilde{P}(i) = i$ for all $i \notin \mathcal{T}.$ Note that $\tilde{P}(\mathcal{T}) = \mathcal{T}$. Now construct a diagonal matrix $\tilde{D} \in \mathbb{F}_q^{*n\times n}$ such that for columns $G_1, \ldots, G_n$ of $G$ we have $G_i = \tilde{D}(i)G_{\tilde{P}^{-1}(i)}$ for all $i \in \mathcal{T}$ and $\tilde{D}(i) = 1$ for all $i \notin \mathcal{T}$. Note that the first restriction on $\tilde{D}$ is always satisfiable because the columns indexed by $\mathcal{T}$ are proportional.

Notice now that $G\tilde{P}\tilde{D} = G$ since for all $i \in \mathcal{T}$
$$(G\tilde{P})_{\tilde{P}(i)} = G_i\implies (G\tilde{P})_i = G_{\tilde{P}^{-1}(i)}\implies (G\tilde{P}\tilde{D})_i = \tilde{D}(i)G_{\tilde{P}^{-1}(i)} = G_i,$$
 and $G_i = (G\tilde{P})_i = (G \tilde{P}\tilde{D})_i$ for all $i \notin \mathcal{T}.$

Now consider the diagonal matrix $D' = P^{-1}\tilde{D}PD$ and permutation matrix $P' = \tilde{P}P$. By definition of $\tilde{P}$, we have that $\forall i \in \mathcal{T}$, $P'(i) = P(\tilde{P}(i)) = P(P^{-1}(\sigma(i))) = \sigma(i)$, and $\forall i\notin\mathcal{T}$, $P'(i) = P(\tilde{P}(i))=P(i)$. Finally, we obtain
$$SGP'D' = SG\tilde{P}PP^{-1}\tilde{D}PD = SG\tilde{P}\tilde{D}PD = SGPD = F.$$

\end{proof}

%Compare this lemma with Lemma III.1. The difference is that we consider proportional rather than equal columns, and this allows us to guarantee a new diagonal $D'$ for the linear equivalence. 
We have shown that any solution $S,P,D$ to $SGPD = F$ may be converted to a new solution $S,P',D'$ where $P'$ bijectively maps, in any chosen way, the indices of one set of proportional columns of $G$ to the indices of the corresponding set of proportional columns of $F$ under the action of $P$. Moreover, we have shown in the proof of Lemma~\ref{lem:newpermutation} that $P'$ could be constructed independently of $D$. This suggests that when searching for a solution $P,D$, we can find $P$ first, and then retrieve $D$. As we shall show, our strategy allows us to first find a valid $P$ by use of an LCE oracle. We then use $P$ to compute $D$.

\begin{comment}
that we may first fix how these column indices are permuted and only then rescale so that the bijection between columns holds. This suggests that when searching for a solution $P,D$, we can find $P$ independently of $D$. As we shall show, we can indeed first find a valid $P$ by use of an oracle that decides LCE. We then use $P$ to manually compute $D$.
\end{comment}

% Permutation finding method/proof begins here:

We denote by $\mathcal{P}(g,G)$ the set of columns in a matrix $G$ that are proportional to a column vector $g,$ and $\mathcal{IP}(g,G)$ the set of indices of all columns in $G$ that are proportional to a column vector $g$. 
%We also abuse the notation for permutations where sometimes it means a permutation of the columns and other times it means a permutation of the indices of the columns. Although, it will be clear which meaning we are using based on what the input to the permutation is (either a number or a column). 
Before we state a key theorem, we need an important remark.
\begin{rmk} \label{rmk:mapsColumnsToColumns}
    Let $G$ and $F$ be $k \times n$ matrices over $\F_q$  that generate linearly equivalent codes via the relation $SGQ = F$ with $Q = PD$ for some diagonal matrix $D$ and permutation matrix $P$. Let $\mathcal{S}$ be a maximal set of indices of proportional columns in $G$ (i.e., if we added the index of another column of $G$ to $\mathcal{S}$, it would no longer be a set indices of proportional columns in $G$). Let $s := |\mathcal{S}|$. Then $\mathcal{S}$ gets mapped to a maximal set of indices $P(\mathcal{S})$ of proportional columns  in $F$ such that $|P(\mathcal{S})| = s$. Furthermore, if $\mathcal{S}$ is the only maximal set of indices of $s$ proportional columns in $G$, and $\mathcal{T}$ is the only maximal set of indices of $s$ proportional columns in $F$, then it must be the case that $P(\mathcal{S}) = \mathcal{T}$.
\end{rmk}

Algorithm~\ref{alg:perm-to-LCE} shows how to compute a permutation matrix $P$ such that there are $S'\in GL_k(\F_q)$ and a diagonal matrix $D$ with $F = S'GPD$ given access to an LCE oracle. The overall strategy is similar to that of Algorithm~\ref{alg:search-to-decision-PCE}. The key difference is that instead of identifying sets of identical columns, we identify sets of columns that are proportional. 

\begin{algorithm}[hbt!]
\caption{Computing $\pi$ with an LCE oracle.}
\begin{algorithmic}[1]\label{alg:perm-to-LCE}
  \REQUIRE Generator matrices $F,G$ of two $[n,k,d]_q$ codes such that there are $S\in GL_k(\F_q)$ and \\
  a monomial matrix $Q=PD$ with $F = SGQ$.
  \ENSURE A permutation matrix $P$ such that there are $S'\in GL_k(\F_q)$ and a diagonal matrix $D$\\ with $F = S'GPD$.
\STATE $m\leftarrow$ maximum number of columns in $G$ that are proportional to each other. 
\STATE $P\leftarrow 0^{n\times n}$, $\mathcal A\leftarrow \{\}$, $\mathcal B\leftarrow\{\}$, $\mathcal{F}\leftarrow\{\}$, $\mathcal{G}\leftarrow \{\}$.
\FOR{$i=1,\ldots,n$}
\IF{$g_i$ is not proportional to an element of $\mathcal{G}$}
\STATE Let $\{G_{t_1},\ldots,G_{t_\ell}\}=\mathcal{G}$, and $\{F_{P(t_1)},\ldots,F_{P(t_\ell)}\}=\mathcal{F}$.
\STATE $G^{(i-1)}\leftarrow G$, $F^{(i-1)}\leftarrow F$.
\FOR{$h\leq \ell$}
\STATE Append $h m$ copies of $G_{t_h}$ at the end of $G^{(i-1)}$.
\STATE Append $h m$ copies of $F_{P(t_h)}$ at the end of $F^{(i-1)}$.
\ENDFOR
\FOR{$j\notin B$}
\STATE $G^{(i)}\leftarrow G^{(i-1)}$, $F^{(i)}\leftarrow F^{(i-1)}$.
\STATE Append $(\ell+1) m$ copies of $G_{i}$ at the end of $G^{(i)}$.
\STATE Append $(\ell+1) m$ copies of $G_{j}$ at the end of $F^{(i)}$.
\IF{$G^{(i)}\stackrel{L}{\approx} F^{(i)}$}
\STATE $\mathcal A\leftarrow\mathcal A\cup \{i\}$, $\mathcal B\leftarrow\mathcal B\cup \{j\}$, $\mathcal{G}\leftarrow \mathcal{G}\cup \{G_i\}$, $\mathcal{F}\leftarrow \mathcal{F}\cup \{F_j\}$, $P_{j,i}\leftarrow 1$.
\STATE \textbf{break}
\ENDIF
\ENDFOR
\ELSE
\STATE Let $t_h$ such that $G_i$ is proportional to $G_{t_h}\in\mathcal{G}$. 
\STATE Find $j\notin B$ such that $F_j$ is proportional to $F_{P(t_h)}$.
\STATE  $\mathcal A\leftarrow \mathcal A\cup \{i\}$, $\mathcal B\leftarrow \mathcal B\cup \{j\}$, $P_{j,i}\leftarrow 1$.
\ENDIF
\ENDFOR
\RETURN $P$
\end{algorithmic}
\end{algorithm}

\begin{thm}
Let $G$ and $F$ be $k \times n$ matrices over $\F_q$ that generate linearly equivalent codes via the relation $SGQ = F$ with $Q = PD.$ Then, Algorithm~\ref{alg:perm-to-LCE}  computes a permutation $P'$ such that there is an invertible matrix $S'$ and a diagonal matrix $D'$ satisfying $F = S'GP'D'$.% by using an oracle for LCE.
\end{thm}
\begin{proof}
 \textbf{First iteration:} We show how to compute the image of one element of $\{1,\cdots, n\}$ via a $P$ that satisfies $SGPD=F$ for some diagonal matrix $D$ by using a decisional oracle. Without loss of generality, we need to compute the image of $1.$ Let $m$ be the maximum number of columns in $G$ that are proportional to each other. Let $i_1 \in \{1,\cdots,n\}$. We want to test if there exist $P$ and $D$ satisfying $SGPD=F$ and $P(1)=i_1$. Let $G_1, G_2, \cdots, G_n$ be the columns of $G$ and let $F_1, F_2, \cdots, F_n$ be the columns of $F$. Create a new matrix $G^{(1)}$ by appending $m$ copies of $G_1$ to $G$ and create a new matrix $F^{(1)}$ by appending $m$ copies of $F_{i_1}$ to $F.$ We ask the LCE oracle if the codes generated by $G^{(1)}$ and $F^{(1)}$ are linearly equivalent. 
\begin{itemize} 
\item If the oracle says "no", then there is no monomial matrix $PD$ realizing the linear equivalence between the codes generated by $G$ and $F$ such that $P(1)=i_1$. Indeed, if there was such a linear equivalence between $G$ and $F$, then one could extend this linear equivalence to a linear equivalence between $G^{(1)}$ and $F^{(1)}$ by creating an $(n+m) \times (n+m)$ permutation $P^\prime$ matrix and diagonal matrix $D^\prime$ such that $P^\prime(i) = P(i)$ and $D^\prime(i) = D(i)$ for $i \in \{1, \cdots, n\},$ and $P^\prime(j) = j$ and $D^\prime(j) = D(i_1)$ for $j \in \{n+1,\cdots, n+m\}.$ But this equivalence between $G^{(1)}$ and $F^{(1)}$ should not exist according to the LCE oracle.
\item If the oracle says "yes", then let $P^\prime$ and $D^\prime$ be a permutation and diagonal matrix satisfying the linear equivalence between $G^{(1)}$ and $F^{(1)},$ and let $S^\prime$ be an invertible matrix such that $S^\prime G^{(1)}P^\prime D^\prime = F^{(1)}.$ Let $\ell:=|\mathcal{P}(G_1,G)|$, $X := \mathcal{P}(G_1,G^{(1)})$ and $Y := \mathcal{P}(F_{i_1},F^{(1)}).$ Because $X$ is the only maximal set of $\ell+m$ pairwise-proportional columns in $G^{(1)}$ and $Y$ is the only maximal set of $\ell+m$ pairwise-proportional columns in $F^{(1)}$, by Remark~\ref{rmk:mapsColumnsToColumns} we have that $P^\prime(X) = Y$. Let $\sigma: \mathcal{IP}(G_1,G^{(1)}) \to \mathcal{IP}(F_{i_1},F^{(1)})$ be a permutation such that $\sigma(1) = i_1$ and $\sigma(j) = j$ for all $j \in \{n+1, \cdots, n+m\}$. By Lemma~\ref{lem:newpermutation} applied to $\sigma$ and $\mathcal{L} = \mathcal{IP}(F_{i_1},F^{(1)})$, there exist a permutation $P'^{(1)},$ a diagonal matrix $D'^{(1)},$ and an invertible matrix $S$ such that $SG^{(1)}P'^{(1)}D'^{(1)} = F^{(1)}$, $P'^{(1)}(1)= i_1$, and $P'^{(1)}|_{\mathcal{IP}(G_1,G^{(1)})} = \sigma$. It follows that there is a linear equivalence given by $S$, the permutation $P^{(1)} := P'^{(1)}|_{\{1,\cdots, n\}}\in\mathcal{S}_n$, and $n \times n$ diagonal matrix $D^{(1)} := \diag(D'^{(1)}(1),\cdots, D'^{(1)}(n))$. Hence, we have shown the existence a monomial matrix $P^{(1)}D^{(1)}$ that defines a linear isometry between $G$ and $F$ such that $P^{(1)}(1) = i_1$. 
\end{itemize}
  \textbf{$T^{th}$ iteration:} Suppose by induction that we have established the existence of a monomial matrix $Q^{(T-1)}=P^{(T-1)}D^{(T-1)}$ where $D^{(T-1)}$ is diagonal and $P^{(T-1)}$ is a permutation matrix such that $S^{(T-1)}G Q^{(T-1)}=F^{(T-1)}$ for $S^{(T-1)}\in\operatorname{GL}_k(\F_q)$, and that the images of $T-1$ elements in the set $\{1, \cdots, n\}$ under the permutation $P^{(T-1)}$ are known. Let $\mathcal A_{T-1}$ be the set of these $T-1$ elements. Let $$\mathcal{G}_{\mathcal A_{T-1}} := \{G_i : i \in \mathcal A_{T-1} \text{ and } \forall i \ne j \in \mathcal A_{T-1}, G_i \notin \mathcal{P}(G_j,G)\}.$$ 
  That is, $\mathcal{G}_{\mathcal A_{T-1}} = \{G_{t_1}, G_{t_2}, \cdots, G_{t_\ell}\}$ is the set of nonproportional columns in $G$ whose image under $P^{(T-1)}$ is known. Similarly define $\mathcal{F}_{\mathcal A_{T-1}} := \{F_{P(t_1)}, F_{P(t_2)},\cdots, F_{P(t_\ell)}\}.$ Note that we need to consider the nonproportional columns to ensure that the number of repetitions per appended column is unique. Let $t \in \{1,2,\cdots, n\} \setminus \mathcal A_{T-1}.$ We now show how to find the image $i_t$ of $t$ via a new permutation $P^{(T)}\in\mathcal{S}_n$ such that $P^{(T)}(i) = P^{(T-1)}(i)$ for $i \in \mathcal A_{T-1}$ and such that there is an invertible matrix $S^{(T)}$ and a diagonal matrix $D^{(T)}$ satisfying $F = S^{(T)}G Q^{(T)}$ for $Q^{(T)} = P^{(T)}D^{(T)}$.
  
\textbf{Case 1:} Suppose that $G_t \in \mathcal{P}(G_i,G)$ for some  $i \in \mathcal A_{T-1}.$ Observe that $\mathcal{IP}(G_i,G) \setminus \mathcal A_{T-1}$ is nonempty because $t$ belongs to this set. Therefore, there exists an $F_{i_t} \in \mathcal{P}(F_{P^{(T-1)}(i)},F)$ such that $i_t \notin P(\mathcal A_{T-1})$. Indeed, by Remark~\ref{rmk:mapsColumnsToColumns}, since $\mathcal{IP}(G_i,G)$ and $\mathcal{IP}(F_{P^{(T-1)}(i)},F)$ are maximal sets of proportional columns of $G$ (resp. $F$), we have $P(\mathcal{IP}(G_i,G)) = \mathcal{IP}(F_{P^{(T-1)}(i)},F)$. Hence, $\mathcal{IP}(F_{P^{(T-1)}(i)},F) \setminus P(\mathcal A_{T-1})\neq\emptyset$. Consider a permutation 
%$$\sigma: \mathcal A_{T-1} \cup \mathcal{IP}(G_i,G) \to P(\mathcal A_{T-1}) \cup \mathcal{IP}(F_{P^{(T-1)}(i)},F)$$ 
$$\sigma: \mathcal{IP}(G_i,G) \to  \mathcal{IP}(F_{P^{(T-1)}(i)},F)$$ 
such that $\sigma(t) = i_t$ and $\sigma(j) = P(j)$ for all $j \in \mathcal A_{T-1}\cap  \mathcal{IP}(G_i,G).$ We can apply Lemma~\ref{lem:newpermutation} to $P^{(T-1)}$, $D^{(T-1)}$, $\sigma$ to prove the existence of a permutation $P^{(T)}$ such that $P^{(T)}(i) = P^{(T-1)}(i)$ for all $i \in \mathcal A_{T-1}$ and such that $P^{(T)}(t) = i_t,$ with $S^{(T-1)}GP^{(T)}D^{(T)} = F$ for some diagonal matrix $D^{(T)}.$ Observe that in this case, one does not need to call an oracle for LCE.\newline
 
\textbf{Case 2:} Suppose that $G_t$ is not proportional to any column in the set $\mathcal{G}_{T-1}.$ Let $i_t \in \{1,\cdots, n \} \setminus P(\mathcal A_{T-1}).$ As before, let $m$ be the maximum number of columns in $G$ that are proportional to each other. Create a new matrix $G^{(T)}$ by appending to $G$: $m$ copies of $G_{t_1},$ $2m$ copies of $G_{t_2}$, $\cdots$, $h\ell m$ copies of $G_{t_\ell},$ and $(\ell+1)m$ copies of $G_t$ and similarly create a new matrix $F^{(T)}$ by appending to $F$: $m$ copies of $F_{P(t_1)},$ $2m$ copies of $F_{P(t_2)}$, $\cdots$, $\ell m$ copies of $F_{P(t_\ell)},$ and $(\ell+1)m$ copies of $F_{i_t}.$ Observe that we have appended $\displaystyle{\sum_{h=1}^{\ell+1} hm = \frac{m(\ell+1)(\ell+2)}{2}}$ columns to $G.$ Now, we ask the LCE oracle if the codes generated by $G^{(T)}$ and $F^{(T)}$ are linearly equivalent. 

 If the oracle says "no", then $P^{(T)}(t) \ne i_t$ for any permutation matrix that defines a linear equivalence between $G$ and $F$ and that satisfies $P^{(T)}(i) = P^{(T-1)}(i)$ for all $i \in \mathcal A_{T-1}.$ If there were such a linear equivalence relation between $G$ and $F$, i.e. $F = S^{(T)}GP^{(T)}D^{(T)}$ for an invertible matrix $S^{(T)}$, a permutation matrix $P^{(T)}$ and a diagonal matrix $D^{(T)}$, then one could extend this linear equivalence to a linear equivalence between $G^{(T)}$ and $F^{(T)}$ by creating a permutation matrix $P'^{(T)}\in\Sn_{N_T}$ (where $N_T$ is the number of columns of $G^{(T)}$) and an $N_T\times N_T$ diagonal matrix $D'^{(T)}$ such that $P'^{(T)}(i) = P^{(T)}(i)$ and $D'^{(T)}(i) = D^{(T)}(i)$ for $i \in \{1, \cdots, n\},$ $P'^{(T)}(j) = j$ for all $j > n,$ and for all $i \in \{0,1,\cdots, \ell\}$ we have that $D'^{(T)}(j) = D^{(T)}(t_i)$ for all $\displaystyle{j \in \left\{n+\frac{mi(i+1)}{2}+1,\cdots, n+\frac{m(i+1)(i+2)}{2}\right\}}.$ This contradicts the assumption that no linear equivalence relation exists between $G^{(T)}$ and $F^{(T)}$.
 
If the oracle says "yes", then there exist an invertible matrix $S'^{(T)}$, a permutation matrix $P'^{(T)}\in\Sn_{N_T}$, and a diagonal matrix $D'^{(T)}$ such that 
$S^{(T)}G^{(T)}P'^{(T)}D'^{(T)} = F^{(T)}$. Let $P^{(T-1)}\in\Sn_n$ such that there are an invertible $S^{(T-1)}$ and a diagonal matrix $D^{(T-1)}$ with $F = S^{(T-1)}GP^{(T-1)}D^{(T-1)}$ and for which we know the images at the indices in $\mathcal A_{T-1}$. We want to show that $P^{(T-1)}$ can be extended to a $P^{(T)}\in\Sn_n$ such that $P^{(T)}(t)=i_t$ and $F = S^{(T)}GP^{(T)}D^{(T)}$ for an invertible matrix $S^{(T)}$ and a diagonal matrix $D^{(T)}$. As in the case of the search to decision reduction for PCE, we need to argue that there is a $P'^{(T)}$ defining the linear equivalence relation between $G^{(T)}$ and $F^{(T)}$ satisfying the following properties: 
\begin{itemize}
    \item it agrees with $P^{(T-1)}$ on $\mathcal A_{T-1}$,
    \item it maps $t$ to $i_t$,
    \item it maps $\{1,\ldots,n\}$ to $\{1,\ldots,n\}$.
\end{itemize}
If such a $P'{(T)}$ exists, then its restriction to $\Sn_n$ gives us $P^{(T)}$, and the diagonal matrix $D^{(T)} = \diag(D'{(T)}(1),\ldots,D'^{(T)}(n))$ satisfies the desired linear equivalence relation. 

Observe that by construction of $G^{(T)},$ for any $h \in \{1,2,\cdots, \ell\}$, we have that $s_h := |\mathcal{P}(G_{t_h},G^{(T)})|= |\mathcal{P}(F_{P'^{(T)}(t_h)},F^{(T)})|$, and that $s_h$ is unique (in the sense that no other maximal set of proportional columns in $G^{(T)}$ has size $s_h$ and no other maximal set of proportional columns in $F^{(T)}$ has size $s_h$). In particular, according to Remark~\ref{rmk:mapsColumnsToColumns}, this means that we must have 
$P'^{(T)}(\mathcal{IP}(G_{t},G^{(T)}) = \mathcal{IP}(F_{P'^{(T)}(t)},F^{(T)}) = \mathcal{IP}(F_{i_t},F^{(T)})$. By applying Lemma~\ref{lem:newpermutation} to $P'^{(T)}$, $D'^{(T)}$, and $\sigma$ defined below, we can assume that 
\begin{itemize}
    \item $P'^{(T)}(t)=i_t$
    \item $P'^{(T)}(h) = h$ for $h\in\{N_T-(\ell+1)m+1,\ldots,N_T\}$ (i.e. $h$ is in the block of the last $(\ell+1)m$ indices).
    \item $P'^{(T)}(h) = \sigma(h)$ for any bijection 
    $$
    \sigma: \{\text{indices of columns of $G$ equal to $G_t$}\}\setminus \{t\}\rightarrow \{\text{indices of columns of $F$ equal to $F_{i_t}$}\}\setminus\{i_t\}.
    $$
\end{itemize}
Moreover, such permutation $P'^{(T)}$ maps $\{1,\ldots,N_T-(\ell+1)m\}$ to $\{1,\ldots,N_T-(\ell+1)m\}$. Hence, its restriction to $\{1,\ldots,N-(\ell+1)m\} = \{1,\ldots,N_{T-1}\}$ is a permutation 
$P'^{(T-1)}\in\Sn_{N_{T-1}}$ satisfying $F^{(T-1)} = S'^{(T-1)}G^{(T-1)}D'^{(T-1)}P'^{(T-1)}$ for $S'^{(T-1)} = S'^{(T)}$ and $D'^{(T-1)} = \diag(D'^{(T)}(1),\ldots,D'^{(N_{T-1})})$. By induction, we know that $P'^{(T-1)}$ can be assumed to agree on $\mathcal{A}_{T-1}$ with $P^{(T-1)}$. We define $P^{(T)}$ as the restriction of $P'^{(T)}$ to $\{1,\ldots,n\}$. It defines a linear equivalence relation between $G$ and $F$, and it satisfies $P^{(T)}(t_h) = i_h$ for all $t_h\in\mathcal A_T$. By induction, the statement hold for all iterations.

\end{proof}

\begin{corollary}\label{cor:poly_perm_lce}
Algorithm~\ref{alg:perm-to-LCE} performs at most $n^2$ calls to an LCE oracle with input codes of dimension $k$ and length at most $n^2(n+1)/2$ over $\F_q$ and eventually returns a permutation matrix $P$ such that there exists an invertible matrix $S$ and a diagonal matrix $D$ with $SGPD=F$. 
\end{corollary}

%% file: Multipliers.tex
\section{Finding the diagonal entries}\label{sec:diagonal-entries}

Let $\mathcal C,\mathcal D$  be two linear codes of length $n$ over $\F_q$ such that there exists a linear isometry $\tau = (\pi,v)\in \mathcal{S}_n\rtimes \F_q^{*n}$ such that $\mathcal D = \tau (\mathcal C)$. Given $\pi$, we show how to efficiently compute $v'\in \F_q^{*n}$ such that 
$\tau' = (\pi,v')\in  \mathcal{S}_n\rtimes \F_q^{*n}$ satisfies $\tau'(\mathcal C) = \mathcal D$. Our main result in this section is Theorem~\ref{thm:diag-coeffs}, which gives a deterministic polynomial-time algorithm for recovering the diagonal part once the permutation is known. This task arises for example in algorithms for solving search-LCE that build on Leon's search-PCE problem~\cite{Leon1982}. Beullens~\cite[Sec. 4]{Beullens20} adapted it to the case of search-LCE and used a backtracking method to recover the linear isometry $\tau$ from several pairs of the form $(V,\tau(V))$ where $V$ is a 2-dimensional subcode. Later,  Barenghi, Biasse, Persichetti and Santini~\cite[Sec. 5.2]{CEcrypto} reformulated Beullen's approach into an explicit two-step process: first compute $\pi\in\mathcal{S}_n$, then derive $v\in\F_q^{*n}$. The approach is given as follows: assume $G\in\F_q^{k\times n}$ is a generator matrix for $\mathcal C$, and $H\in\F_{q}^{(n-k)\times n}$ is a parity-check matrix for $\mathcal{D}$. Assume $P$ is a known permutation matrix and $D$ an unknown diagonal matrix such that $PD$ corresponds to $\tau$. They find the $n$ unknown coefficients of $D$ by noticing that the identity 
$$
G PD H = 0
$$
induces $(n-k)n$ linear equations between them. This approach is given from a very high level perspective in~\cite[Sec. 5.2]{CEcrypto} because it is considered folklore. However, to the best of our knowledge, its careful analysis does not appear anywhere in the literature, and it appears to be heuristic. We recall this method for completeness. 

The purpose of this section is twofold. We first recall the standard approach,
which formulates the recovery of the diagonal entries as a linear system derived
from the identity $GPDH = 0$. While this method is natural from a linear-algebraic
perspective, it is in general only heuristic, as the resulting system may admit
multiple solutions and does not guarantee nonzero diagonal entries. We then discuss a simpler deterministic strategy in the presence of a projective
frame, where proportionality constraints uniquely determine the linear action. We provide some intuition for this approach. In a projective reconstruction,
one typically relies on a column whose coordinates are all nonzero, allowing
all scaling factors to be determined simultaneously. In the absence of such a frame, the proportionality constraints do not determine
the linear action uniquely and may admit degenerate solutions. In contrast, our method
does not require such a column. Instead, it exploits overlaps between the
supports of multiple columns: each column induces relations between scaling
factors on its support, and when these supports form a connected structure,
these local constraints propagate across all coordinates. In this way, the diagonal entries can be recovered consistently without requiring
a projective frame, by propagating local constraints through the support
structure.
%Our main contribution is a different deterministic method, based on row-reduced
%echelon forms and connectivity properties of column supports, which bypasses the
%need for a projective frame and yields a polynomial-time recovery of the diagonal
%part in full generality.

Crucially, one has to identify a solution for $D$ whose coefficients are non-zero. It is typically ensured when the system is overdetermined and has a unique solution. However, in general, we may only assume that it has rank $t\leq n$. The diagonal matrices derived from the space of solutions have the shape $D = \diag(x_1,\ldots,x_t,h_1(\overline{x}),\ldots,h_{n-t}(\overline{x}))$ where $\overline{x} = x_1,\ldots,x_{t}$ and where the $h_i$ are linear functions. Let 
$$
P(x_1,\ldots,x_t) = \left( \prod_{i\leq t}x_i \right)\cdot \left(\prod_{j \leq n-t}h_i(x_1,\ldots,x_t)\right).
$$
The task of finding a solution $D$ with non-zero coefficients reduces to the task of finding $a_1,\ldots,a_i\in\F_q^{*t}$ such that $P(a_1,\ldots,a_t)\neq 0$. If the degree of $P$ in a variable $x_i$ is greater than $q$, we substitute $x_i^q$ by $x_i$ to reduce our task to the search for $(a_1,\ldots,a_t)$ such that $R(a_1,\ldots,a_t)\neq 0$ where $R$ is the reduction of $P$ modulo $I:=\langle x_1^q-x_1,\ldots,x_t^q-x_t\rangle$. In the absence of additional structural assumptions, this step is naturally
viewed as heuristic. Empirically, for typical instances, a small number of
random choices of $(a_1,\dots,a_t)$ suffices to obtain a nonzero evaluation.

A possible rigorous and deterministic way to derive the unknown diagonal coefficients would be to first search for $S\in\mathrm{GL}_k(\mathbb{F}_q)$ such that $F = SGPD$ where $G$ is a generator matrix for $\mathcal C$ and $F$ is a generator matrix for $\mathcal D$ (given $P$ found via the methods of Section~\ref{sec:LEC-perm-search}). Hence, $S$ satisfies the collinearity constraints on the columns of $F$ and $G'$ $F_i\sim SG'_i$ $\forall i\leq n$ where $G' = GP$. If the columns of $G'$ (and hence of $F$) contain a projective frame, Lemma~\ref{lem:pgl-uniqueness} shows that any non-zero solution for $S$ will be invertible (and hence lead to a solution for $D$). The search for $S$ can be linearized by turning each condition $F_i\sim SG'_i$ into $\binom{k}{2}$ homogeneous linear equations given by $(F_i)_r(SG'_i)_s = (F_i)_s(SG'_i)_r$ for each pair of indices $1\leq r<s\leq k$. Solutions to the resulting linear system are exactly the matrices $S$ that satisfy $F = SG'D$ for some diagonal matrix $D$, and from Lemma~\ref{lem:pgl-uniqueness}, a non-zero solution must be invertible.

\begin{lemma}[Projective agreement on a frame forces equality up to scalar]\label{lem:pgl-uniqueness}
Let $g_1,\dots,g_n \in \F_q^k$ be nonzero vectors. Assume that among the projective points
$[g_i]\in \mathbb{P}^{k-1}(\F_q)$ there exists a \emph{projective frame}, i.e., there are indices
$i_0,i_1,\dots,i_k$ such that $g_{i_1},\dots,g_{i_k}$ are linearly independent and, when writing
\[
g_{i_0}=\sum_{j=1}^k a_j g_{i_j}\quad\text{in that basis,}
\]
all coefficients satisfy $a_j\neq 0$.

Suppose there exists $S_{\text{true}}\in \mathrm{GL}_k(\F_q)$ and nonzero scalars $d_1,\dots,d_n\in \F_q^{*}$ such that
\[
g_i' = d_i\, S_{\text{true}}\, g_i\qquad\text{for all } i=1,\dots,n.
\]
Let $S\in \F_q^{k\times k}$ be any matrix satisfying the projective constraints
\[
S g_i \sim g_i' \qquad\text{for all } i=1,\dots,n,
\]
i.e., for each $i$ there exists $\lambda_i\in \F_q$ with $S g_i = \lambda_i g_i'$.
Then either $S=0$, or $S$ is invertible and there exists $\alpha\in \F_q^{*}$ such that
\[
S = \alpha\, S_{\text{true}}.
\]
Equivalently, the induced projective transformation $[S]\in \mathrm{PGL}_k(\F_q)$ (when defined) is unique and
equals $[S_{\text{true}}]$.
\end{lemma}

\begin{proof}
For each $i$ we have $S g_i = \lambda_i g_i' = (\lambda_i d_i)\, S_{\text{true}} g_i$.
Define $T := S_{\text{true}}^{-1} S \in M_k(\F_q)$. Then for all $i$,
\[
T g_i = \mu_i g_i\qquad\text{with }\mu_i := \lambda_i d_i \in \F_q.
\]
In particular, for $j=1,\dots,k$ we have $T g_{i_j} = \mu_{i_j} g_{i_j}$.
Since $g_{i_1},\dots,g_{i_k}$ are a basis, $T$ is diagonal in this basis:
for any $x=\sum_{j=1}^k x_j g_{i_j}$,
\[
T x = \sum_{j=1}^k x_j \mu_{i_j} g_{i_j}.
\]
Now apply this to $g_{i_0}=\sum_{j=1}^k a_j g_{i_j}$ with all $a_j\neq 0$:
on one hand $T g_{i_0} = \mu_{i_0} g_{i_0} = \sum_{j=1}^k (\mu_{i_0} a_j) g_{i_j}$,
and on the other hand
\[
T g_{i_0}=\sum_{j=1}^k a_j \mu_{i_j} g_{i_j}.
\]
Comparing coefficients in the basis $(g_{i_j})_{j=1}^k$ gives $a_j \mu_{i_j} = a_j \mu_{i_0}$ for all $j$,
hence $\mu_{i_j}=\mu_{i_0}$ for all $j$ (since $a_j\neq 0$).
Therefore $T = \mu_{i_0} I_k$ is a scalar matrix.

If $\mu_{i_0}=0$ then $T=0$ and thus $S=0$. Otherwise $T$ is invertible, hence $S$ is invertible and
$S = \mu_{i_0} S_{\text{true}}$.
\end{proof}

\begin{rmk}[Why the frame hypothesis is needed]
If the set $\{[g_i]\}_{i\leq n}$ does \emph{not} contain a projective frame, the conclusion can fail.
For example, if the columns only represent the coordinate axes
$g_1=e_1,\dots,g_k=e_k$, then any diagonal matrix
$T=\mathrm{diag}(c_1,\dots,c_k)$ satisfies $T g_i \sim g_i$ for all $i$,
without being a scalar multiple of the identity when the $c_j$ are not all equal.
\end{rmk}

To circumvent this restriction, we replace the projective frame requirement with a weaker, global condition based on
the supports of the columns. Rather than relying on a single column, we exploit overlaps between the supports of multiple columns. Each column relates the scaling factors on the coordinates in its support, and when these supports form a connected structure, these local relations propagate across all coordinates. In this way, the diagonal entries can be recovered consistently without ever requiring a column with full support.

\begin{lemma}\label{lem:DRREFD}
Let $A \in \mathbb{F}_q^{k \times n}$ be a full-rank matrix with $n > k$ and $D \in \mathbb{F}_q^{*n \times n}$ be a diagonal matrix. Then there exists $D' \in \mathbb{F}_q^{*k \times k}$ such that $RREF(AD) = D'RREF(A)D$.
\end{lemma}

\begin{proof}
Let $S \in GL_k(\mathbb{F}_q)$ be such that $RREF(A) = S A$. Now consider the matrix $D'\in \mathbb{F}_q^{*k \times k}$ such that $D'(i) = 1/D(p_i)$ for all $i \in \{1,\ldots, k\}$ where $\{p_1,\ldots, p_k\}$ are the indices of the pivot columns of $SA$. Note that there are indeed $k$ pivot columns since $A$ has full rank. Notice then that $D' S A D$ is also in RREF because the pivot columns of $SA$ remain unchanged after multiplying its columns by the elements of $D$ and rows by the elements of $D'$. But then $D'S$ is a non-singular matrix such that $D' S A D$ has reduced row echelon form, and so $D' S A D = RREF(AD)$ by uniqueness. We conclude that $RREF(AD) = D'RREF(A)D$.

\end{proof}

\begin{corollary}\label{cor:Zeros}
Let $A \in \mathbb{F}_q^{k \times n}$ be a full-rank matrix with $n > k$ and $D \in \mathbb{F}_q^{*n \times n}$ be a diagonal matrix. Then $RREF(AD)$ and $RREF(A)$ have $0$'s at the same indices.
\end{corollary}

\begin{proof}
Simply notice that, in the proof of Lemma \ref{lem:DRREFD}, the diagonal elements of $D'$ and $D$ are non-zero. Therefore, each element of $RREF(A)$ differs from its corresponding element of $RREF(AD)$ only by multiplication by a non-zero scalar.

\end{proof}

We can now construct an efficient algorithm for calculating a diagonal matrix $D$ that, with our obtained permutation matrix $P$, composes the linear isometry between our two linear codes. We start by dealing with the special case where the non-pivot columns of $RREF(F)$ have connected support. We formally define this property below: 

\begin{definition}[Connected supports]
We say that column vectors $C_1,\ldots,C_n$ have connected support if there is an ordering $i_1,\ldots,i_n$ of $\{1,\ldots,n\}$ such that if $\mathcal{S}_i$ is the support of $C_i$, we have
$$
\forall 1\leq j\leq n,\ \ \mathcal{S}_{i_j}\cap \bigcup_{0<\ell<j}\mathcal{S}_{i_\ell} \neq \emptyset.
$$
We call $\{i_1,\ldots,i_n\}$ a connected ordering of the columns $C_i$.
\end{definition}

\begin{algorithm}[ht]
\caption{Computation of the diagonal coefficients (connected case).}
\begin{algorithmic}[1]\label{alg:multipliers}
  \REQUIRE Let $F,G\in \F_q^{k\times n}$ be generator matrices of $[n,k,d]_q$ codes such that the non-pivot\\ columns of $RREF(F)$ have connected support, and a permutation matrix $P$ such that\\ there are $S\in GL_k(\F_q)$, and a diagonal matrix $D$ with $F = SGPD$.
  \ENSURE $S\in GL_k(\F_q)$ and a diagonal matrix $D$ with $F = SGPD$.
\STATE Compute $S_1,S_2 \in GL_k(\mathbb{F})$ with $S_1GP = RREF(GP)$ and $S_2F = RREF(F)$.
\STATE $\{i_1,\ldots,i_{n-k}\}\leftarrow$ a connected ordering of the non-pivot columns of $RREF(F)$.
\STATE $D'\leftarrow I_n$, $\mathcal{A}\leftarrow \{\}$.
\FOR{$\ell \in \{i_1,\ldots,i_{n-k}\}$}
\STATE $c\leftarrow$ $\ell$-th column of $RREF(GP)$. $\tilde{c}\leftarrow$ $\ell$-th column of $RREF(F)$. 
\STATE $\mathcal{S}\leftarrow$ support of $c$ (and hence of $\tilde{c}$).
\IF{$\mathcal{A}=\emptyset$}
\STATE $\lambda\leftarrow 1$.
\ELSE
\STATE Let $s\in\mathcal{S}\cap\mathcal{A}$. $\lambda\leftarrow D'_{s,s} c_s/\tilde{c}_s$.
\ENDIF
\FOR{$i\in\mathcal{S}\setminus\mathcal{A}$}
\STATE $D'_{i,i}\leftarrow \lambda \tilde{c}_i/c_i$, $\mathcal{A}\leftarrow\mathcal{A}\cup\{i\}$
\ENDFOR
\ENDFOR
\STATE Let $c'^1,\ldots,c'^n$ be the columns of $D'RREF(GP)$.
\STATE Let $\tilde{c}^1,\ldots,\tilde{c}^n$ be the columns of $RREF(F)$.
\STATE $D\leftarrow I_n$.
\FOR{$i\leq n$}
\STATE Let $j$ be an index satisfying $c'^i_j\neq 0$.
%\IF{there is an index $j$ with $c'^i_j\neq 0$}
\STATE $D_{i,i}\leftarrow \tilde{c}^i_j/c'^i_j$.
%\ENDIF
\ENDFOR
\STATE $S\leftarrow S_2^{-1}D'S_1$.
\RETURN $S,D$
\end{algorithmic}
\end{algorithm}

\begin{lemma}

Let $F,G\in \F_q^{k\times n}$ be generator matrices of $[n,k]_q$-linear codes where the non-pivot columns of $RREF(F)$ have connected support, and such that there exist a non-singular matrix $S \in GL_k(\mathbb{F}_q)$, and permutation matrix $P \in \{0,1\}^{n \times n}$ such that $SGPD = F$ for some diagonal matrix $D \in \mathbb{F}_q^{*n\times n}$. Given $F,G$, and $P$, Algorithm~\ref{alg:multipliers} computes $S,D$ such that $F = SGPD$ in polynomial time.
\end{lemma}

\begin{proof}

% TODO
First observe by Lemma \ref{lem:DRREFD} that
\begin{equation}
RREF(F) = RREF(SGPD) = RREF(GPD) = D'RREF(GP)D\label{eq:6.1}\tag{1}
\end{equation}
\noindent for some diagonal matrix $D' \in \mathbb{F}_q^{*k \times k}$.

The matrix $D' = \diag(D'_{i,i})$ has the property that for all $j$, if $c$ is the $j$-th column of $RREF(GP)$, and $\tilde{c}$ is the $j$-th column of $RREF(F)$, then $D'c\sim \tilde{c}$ (more specifically, $D'c = D_{j,j}\tilde{c}$). Let $\mathcal{S} = \{i_1,\ldots,i_s\}$ be the support of $c$ (and hence $\tilde{c}$ since they have same support). Then for any vector $(x_{i_1},\ldots,x_{i_s})$ we have  
\begin{align*}
\begin{pmatrix}
    x_{i_1}\\
    \vdots \\
    x_{i_{s}}
\end{pmatrix}\sim 
\begin{pmatrix}
    \tilde{c}_{i_1}/c_{i_1}\\
    \vdots\\
    \tilde{c}_{i_s}/c_{i_s}
\end{pmatrix}
&\Leftrightarrow
\exists\lambda\in\F_q^*\ \ 
\left\{
\begin{matrix}
    x_{i_1}c_{i_1} = \lambda\tilde{c}_{i_1} = c_{i_1}\lambda/D_{j,j} D'_{i_1,i_1}\\
    \vdots\\
   x_{i_s}c_{i_s} = \lambda\tilde{c}_{i_s} = c_{i_s}\lambda/D_{j,j} D'_{i_s,i_s} 
\end{matrix}\right.\\
&\Leftrightarrow
\begin{pmatrix}
    x_{i_1}\\
    \vdots \\
    x_{i_{s}}
\end{pmatrix}\sim 
\begin{pmatrix}
    D'_{i_1,i_1}\\
    \vdots\\
    D'_{i_s,i_s}
\end{pmatrix}
\end{align*}

Now, let us show the loop invariant of Steps~4-15. Let $D'_{\text{real}}$ be the diagonal matrix that satisfies $RREF(F) = D'_{\text{real}}RREF(GP)D$, and let $\mathcal{A} = \{i_1,\cdots i_a\}$. Then 
$$\begin{pmatrix}
    D'_{i_1,i_1}\\
    \vdots\\
    D'_{i_a,i_a}
\end{pmatrix}
\sim
\begin{pmatrix}
    D'_{\text{real},i_1,i_1}\\
    \vdots\\
    D'_{\text{real},i_a,i_a}
\end{pmatrix}.$$
At the beginning of the first iteration of the loop, $\mathcal{A}=\emptyset$. Hence $\lambda=1$, and at the end of the iteration, we have 
$$
\begin{pmatrix}
 D'_{i_1}\\
 \vdots\\
 D'_{i_a}
\end{pmatrix}
= 
\begin{pmatrix}
 \tilde{c}_{i_1}/c_{i_1}\\
 \vdots\\
 \tilde{c}_{i_a}/c_{i_a}
\end{pmatrix}
\sim
\begin{pmatrix}
     D'_{\text{real},i_1,i_1}\\
    \vdots\\
    D'_{\text{real},i_a,i_a}   
\end{pmatrix}.
$$
Now assume that we enter an iteration, and that at the beginning of the loop, the invariant is satisfied. Let us show that the loop invariant is still satisfied at the end of the loop. In Step~5, there exists $\lambda_0\in\F_q^*$ such that 
$$
\begin{pmatrix}
 D'_{i_1,i_1}\\
 \vdots\\
 D'_{i_a,i_a}
\end{pmatrix}
= 
\lambda_0
\begin{pmatrix}
     D'_{\text{real},i_1,i_1}\\
    \vdots\\
    D'_{\text{real},i_a,i_a}   
\end{pmatrix}.
$$
 For the index $s$ chosen in Step~10, we have $D'_{\text{real},s}c_s = D_{\ell,\ell}\tilde{c}_s$ (since we know that $D'_{\text{real}}c = D_{\ell,\ell}\tilde{c}$). Note that we know that $\mathcal{S}\cap\mathcal{A}\neq \emptyset$ since $RREF(F)$ has non-pivot columns of connected support. For $i\in\mathcal{S}\setminus\mathcal{A}$, we have $c_s/\tilde{c}_s = D_{\ell,\ell}/D'_{\text{real},s,s}$, hence: 
$$
\lambda = D'_{s,s} c_s/\tilde{c}_s = \lambda_0 D'_{\text{real},s,s}c_s/\tilde{c}_s = \lambda_0 D_{\ell,\ell}.
$$
Then we set $D'_{i,i}$ to 
$$
D'_{i,i} = \lambda\tilde{c}_i/c_i = \lambda D'_{\text{real},i,i}/D_{\ell,\ell} = \lambda_0 D'_{\text{real},i,i}.
$$
Hence, if $\mathcal{S}\setminus\mathcal{A} = \{j_1,\ldots,j_s\}$, at the end of the loop, we have 
$$
\begin{pmatrix}
 D'_{i_1,i_1}\\
 \vdots\\
 D'_{i_a,i_1}\\
 D'_{j_1,j_1}\\
 \vdots\\
 D'_{j_s,j_s}
\end{pmatrix}
= 
\lambda_0
\begin{pmatrix}
     D'_{\text{real},i_1,i_1}\\
    \vdots\\
    D'_{\text{real},i_a,i_a}\\
    D'_{\text{real}j_1,j_1}\\
    \vdots\\
    D'_{\text{real}j_s,j_s}
\end{pmatrix}.
$$
This resulting diagonal matrix $D'$ is such the $i$-th column of $D'RREF(GP)$ is proportional to the $i$-th column of $RREF(F)$. Notice that the pivot columns remain proportional despite multiplication by $D'$. Then Steps~19-22 successfully recover a multiple of the actual $D_{i,i}$ for all $i$. Indeed, $c'^i = D'c^i = \lambda_0 D'_{\text{real}}c^i = \lambda_0 D_{i,i} \tilde{c}^i$ (note that we have excluded zero columns by restricting ourselves to non-pivots of connected support).

We will now show that this multiple of $D$ is an acceptable diagonal matrix for the resolution of our instance of search-LCE. Notice that we now have $D'RREF(GP)D = RREF(F)$. That is, $D'S_1GPD = S_2F$, and so $S_2^{-1}D'S_1GPD = F$. But $S = S_2^{-1}D'S_1$ is a product of non-singular matrices, and is thus non-singular. Therefore, $D$ is such that $SGPD = F$, as desired.

The run time of the algorithm is polynomial in the size of the input. 
\end{proof}

In the case where the non-pivot columns of $RREF(F)$ do not have a connected support the codes generated by $G$ and $F$ are decomposable i.e. they are permutationally equivalent to a direct sum 
\begin{align*}
  \mathcal{C}&\simeq \mathcal{C}_1\oplus\ldots\oplus \mathcal{C}_m\\
  \mathcal{D}&\simeq \mathcal{D}_1\oplus\ldots\oplus \mathcal{D}_m
\end{align*}
for some $m>2$ and indecomposable codes $(\mathcal{C}_i)_{i\leq m}, (\mathcal{D}_i)_{i\leq m}$ satisfying $\mathcal{C}_i\stackrel{L}{\approx}\mathcal{D}_i$ for all $i\leq m$. Standard methods such as~\cite[Alg. 1]{BHJ26} allow the computation of the corresponding permutation and irreducible components in polynomial time. This means that we get permutation matrices $P_\mathcal{C},P_\mathcal{D}$,  invertible matrices $S_\mathcal{C},S_\mathcal{D}$, and generator matrices $G',F'$ of the form 
$$
G' = \begin{pmatrix}
    G^{(1)} & (0) &  \\
        & \ddots & \\
        &    (0)    & G^{(r)}
\end{pmatrix},\qquad 
F' = \begin{pmatrix}
    F^{(1)} & (0) &  \\
        & \ddots & \\
        &    (0)    & F^{(r)}
\end{pmatrix}
$$
such that the $G^{(i)},F^{(i)}$ generate indecomposable codes (in particular: their information sets have connected support), and $G' = S_\mathcal{C} G P_\mathcal{C}$, $F'=S_\mathcal{D} F P_\mathcal{D}$. With the decisional LCE oracle, we can test which $G^{(i)}$ generate a code linearly equivalent to the code generated by an $F^{(j)}$. Hence, after $r^2$ calls to the LCE oracle, we can assume that there exist invertible matrices $(S^{(i)})_{i\leq r}$, permutation matrices $(P^{(i)})_{i\leq r}$ and diagonal matrices $(D^{(i)})$ such that 
$$
F^{(i)} = S^{(i)}G^{(i)}P^{(i)}D^{(i)}
$$
Each $P^{(i)}$ can be recovered using Algorithm~\ref{alg:perm-to-LCE} while we get the $D^{(i)}$ from Algorithm~\ref{alg:multipliers}. This means that we have
\begin{align*}
F &= S^{-1}_\mathcal{D} F' P^{-1}_\mathcal{D} \\
&= S^{-1}_\mathcal{D} \diag(S^{(i)}) G' \diag(P^{(i)}) \diag(D^{(i)}) P^{-1}_\mathcal{D}\\
&=  \underbrace{S^{-1}_\mathcal{D} \diag(S^{(i)}) S_\mathcal{C}}_S G \underbrace{P_{\mathcal{C}} \diag(P^{(i)})}_P\underbrace{ \diag(D^{(i)})P^{-1}_\mathcal{D}}_D,
\end{align*}
where $S$ is invertible, $P$ is a permutation matrix, and $D$ is a diagonal matrix.

\begin{algorithm}[ht]
\caption{Search to Decision reduction for LCE (general case).}
\begin{algorithmic}[1]\label{alg:multipliers-generic}
  \REQUIRE Generator matrices $F,G$ of $[n,k,d]_q$ codes such that there are $S\in GL_k(\F_q)$, \\
  a permutation matrix $P$, and a diagonal matrix $D$ with $F = SGPD$.
  \ENSURE A permutation matrix $P$, a diagonal matrix $D$ and $S\in GL_k(\F_q)$ such that $F = SGPD$.
\STATE Get permutation matrix $P_\mathcal{C}$ and invertible matrix $S_\mathcal{C}$ such that $S_\mathcal{C}GP_\mathcal{C} = \diag(G^{(i)})$\\ where the $(G^{(i)})_{i\leq m}$ generate irreducible codes $(\mathcal{C}_i)_{i\leq m}$.
\STATE Get permutation matrix $P_\mathcal{D}$ and invertible matrix $S_\mathcal{D}$ such that $S_\mathcal{D}FP_\mathcal{D} = \diag(F^{(i)})$\\ where the $(F^{(i)})_{i\leq m}$ generate irreducible codes $(\mathcal{D}_i)_{i\leq m}$.
\FOR{$i,j\leq m$}
\STATE Test if $\mathcal{C}_i\stackrel{L}{\approx}\mathcal{D}_j$.
\ENDFOR
\STATE Rearrange $P_\mathcal{D}$ to ensure that $\forall i\leq m$, $\mathcal{C}_i\stackrel{L}{\approx}\mathcal{D}_i$.
\FOR{$i\leq m$}
\STATE Use Algorithms~\ref{alg:perm-to-LCE} and~\ref{alg:multipliers} to find invertible matrices $(S^{(i)})_{i\leq m}$, permutation matrices $(P^{(i)})_{i\leq m}$, and diagonal matrices $(D^{(i)})_{i\leq m}$ such that $
F^{(i)} = S^{(i)}G^{(i)}P^{(i)}D^{(i)}
$
\ENDFOR
\STATE $S\leftarrow S^{-1}_\mathcal{D} \diag(S^{(i)}) S_\mathcal{C}$.
\STATE $P\leftarrow P_{\mathcal{C}} \diag(P^{(i)})$.
\STATE $D\leftarrow \diag(D^{(i)})P^{-1}_\mathcal{D}$
\RETURN $S,P,D$
\end{algorithmic}
\end{algorithm}

\begin{thm}\label{thm:diag-coeffs}
    Let $F,G\in \F_q^{k\times n}$ be generator matrices of $[n,k]_q$-linear codes, and such that there exist a non-singular matrix $S \in GL_k(\mathbb{F}_q)$, and permutation matrix $P \in \{0,1\}^{n \times n}$ such that $SGPD = F$ for some diagonal matrix $D \in \mathbb{F}_q^{*n\times n}$. Given $F,G$, Algorithm~\ref{alg:multipliers-generic} computes $S,P,D$ such that $F = SGPD$. in polynomial time.
\end{thm}

Together with Corollary \ref{cor:poly_perm_lce}, Theorem \ref{thm:diag-coeffs} completes the polynomial-time search-to-decision reduction for LCE: given access to a decisional LCE oracle, one can recover a full linear isometry $\tau = (\pi, \nu)$ between any two linearly equivalent codes in polynomial time. This places search-LCE and decisional LCE in the same complexity class under polynomial-time reductions.

%% file: ResearchOutline.bib
@INPROCEEDINGS{Reduction,
  author={Biasse, Jean-François and Micheli, Giacomo},
  booktitle={2023 IEEE International Symposium on Information Theory (ISIT)}, 
  title={A Search-to-Decision Reduction for the Permutation Code Equivalence Problem}, 
  year={2023},
  volume={},
  number={},
  pages={602-607},
  doi={10.1109/ISIT54713.2023.10206940}}

@article{mceliece1978public,
	title={A public-key cryptosystem based on algebraic coding theory},
	author={R. McEliece},
	journal={DSN progress report},
	volume={42},
	number={44},
	pages={114--116},
	year={1978}
}

@InProceedings{Girault,
author="M. Girault",
editor="J. Seberry
and J. Pieprzyk",
title="A (non-practical) three-pass identification protocol using coding theory",
booktitle="Advances in Cryptology --- AUSCRYPT '90",
year="1990",
publisher="Springer Berlin Heidelberg",
address="Berlin, Heidelberg",
pages="265--272",
isbn="978-3-540-46297-2"
}

@InProceedings{Sendrier-Simos,
author="N. Sendrier 
and D. Simos",
editor="P. Gaborit",
title="The Hardness of Code Equivalence over $\mathbb{F}_q$ and Its Application to Code-Based Cryptography",
booktitle="Post-Quantum Cryptography",
year="2013",
publisher="Springer Berlin Heidelberg",
address="Berlin, Heidelberg",
pages="203--216",
isbn="978-3-642-38616-9"
}

@inproceedings{LESS,
  author    = {J.{-}F. Biasse and
               G. Micheli and
               E. Persichetti and
               P. Santini},
  editor    = {A. Nitaj and
               A. Youssef},
  title     = {{LESS} is More: Code-Based Signatures Without Syndromes},
  booktitle = {Progress in Cryptology - {AFRICACRYPT} 2020 - 12th International Conference
               on Cryptology in Africa, Cairo, Egypt, July 20-22, 2020, Proceedings},
  series    = {Lecture Notes in Computer Science},
  volume    = {12174},
  pages     = {45--65},
  publisher = {Springer},
  year      = {2020},
  url       = {https://doi.org/10.1007/978-3-030-51938-4\_3},
  doi       = {10.1007/978-3-030-51938-4\_3},
  bibsource = {dblp computer science bibliography, https://dblp.org}
}

@article{Couveignes06,
  author    = {J.-M. Couveignes},
  title     = {Hard Homogeneous Spaces},
  journal   = {{IACR} Cryptol. ePrint Arch.},
  pages     = {291},
  year      = {2006},
  url       = {http://eprint.iacr.org/2006/291},
  bibsource = {dblp computer science bibliography, https://dblp.org}
}

@inproceedings{AlamatiFMP20,
  author    = {N. Alamati and
               L. De Feo and
               H. Montgomery and
               S. Patranabis},
  editor    = {S. Moriai and
               H. Wang},
  title     = {Cryptographic Group Actions and Applications},
  booktitle = {Advances in Cryptology - {ASIACRYPT} 2020 - 26th International Conference
               on the Theory and Application of Cryptology and Information Security,
               Daejeon, South Korea, December 7-11, 2020, Proceedings, Part {II}},
  series    = {Lecture Notes in Computer Science},
  volume    = {12492},
  pages     = {411--439},
  publisher = {Springer},
  year      = {2020},
  url       = {https://doi.org/10.1007/978-3-030-64834-3\_14},
  doi       = {10.1007/978-3-030-64834-3\_14},
  bibsource = {dblp computer science bibliography, https://dblp.org}
}

@inproceedings{CSIDH,
  author    = {W. Castryck and
               T. Lange and
               C. Martindale and
               L. Panny and
               J. Renes},
  editor    = {T. Peyrin and
               S. Galbraith},
  title     = {{CSIDH:} An Efficient Post-Quantum Commutative Group Action},
  booktitle = {Advances in Cryptology - {ASIACRYPT} 2018 - 24th International Conference
               on the Theory and Application of Cryptology and Information Security,
               Brisbane, QLD, Australia, December 2-6, 2018, Proceedings, Part {III}},
  series    = {Lecture Notes in Computer Science},
  volume    = {11274},
  pages     = {395--427},
  publisher = {Springer},
  year      = {2018},
  url       = {https://doi.org/10.1007/978-3-030-03332-3\_15},
  doi       = {10.1007/978-3-030-03332-3\_15},
  bibsource = {dblp computer science bibliography, https://dblp.org}
}

@inproceedings{LESS-FM,
  author    = {A. Barenghi and
               J.{-}F. Biasse and
               E. Persichetti and
               P. Santini},
  editor    = {J. H. Cheon and
               J.{-}P. Tillich},
  title     = {{LESS-FM:} Fine-Tuning Signatures from the Code Equivalence Problem},
  booktitle = {Post-Quantum Cryptography - 12th International Workshop, PQCrypto
               2021, Daejeon, South Korea, July 20-22, 2021, Proceedings},
  series    = {Lecture Notes in Computer Science},
  volume    = {12841},
  pages     = {23--43},
  publisher = {Springer},
  year      = {2021},
  url       = {https://doi.org/10.1007/978-3-030-81293-5\_2},
  doi       = {10.1007/978-3-030-81293-5\_2},
  bibsource = {dblp computer science bibliography, https://dblp.org}
}

@article{LESS-advanced,
  author    = {A. Barenghi and
               J.{-}F. Biasse and
               T. Ngo and
               E. Persichetti and
               P. Santini},
  title     = {Advanced signature functionalities from the code equivalence problem},
  journal   = {Int. J. Comput. Math. Comput. Syst. Theory},
  volume    = {7},
  number    = {2},
  pages     = {112--128},
  year      = {2022},
  url       = {https://doi.org/10.1080/23799927.2022.2048206},
  doi       = {10.1080/23799927.2022.2048206},
  bibsource = {dblp computer science bibliography, https://dblp.org}
}

@ARTICLE{Leon1982,
author={J. {Leon}},
journal={IEEE Transactions on Information Theory},
title={Computing automorphism groups of error-correcting codes},
year={1982},
volume={28},
number={3},
pages={496-511},
month={May},}

@inproceedings{Beullens20,
  author    = {W. Beullens},
  editor    = {O. Dunkelman and
               M. Jacobson Jr. and
               C. O'Flynn},
  title     = {Not Enough {LESS:} An Improved Algorithm for Solving Code Equivalence
               Problems over $\mathbb{F}_q$},
  booktitle = {Selected Areas in Cryptography - {SAC} 2020 - 27th International Conference,
               Halifax, NS, Canada (Virtual Event), October 21-23, 2020, Revised
               Selected Papers},
  series    = {Lecture Notes in Computer Science},
  volume    = {12804},
  pages     = {387--403},
  publisher = {Springer},
  year      = {2020},
  url       = {https://doi.org/10.1007/978-3-030-81652-0\_15},
  doi       = {10.1007/978-3-030-81652-0\_15},
  bibsource = {dblp computer science bibliography, https://dblp.org}
}

@article{CEcrypto,
 author    = {A. Barenghi and
               J.{-}F. Biasse and
               E. Persichetti and
               P. Santini},
journal={Advances in Mathematics of Communications},
title={On the computational hardness of the code equivalence problem in cryptography},
year={2022},
volume={17},
number={1},
pages={23-55}
}

@article{Sendrier-SSA,
  author    = {N. Sendrier},
  title     = {Finding the permutation between equivalent linear codes: The support
               splitting algorithm},
  journal   = {{IEEE} Trans. Inf. Theory},
  volume    = {46},
  number    = {4},
  pages     = {1193--1203},
  year      = {2000},
  url       = {https://doi.org/10.1109/18.850662},
  doi       = {10.1109/18.850662},
  bibsource = {dblp computer science bibliography, https://dblp.org}
}

@inproceedings{BardetOS19,
  author    = {M. Bardet and
               A. Otmani and
               M. Saeed{-}Taha},
  title     = {Permutation Code Equivalence is Not Harder Than Graph Isomorphism
               When Hulls Are Trivial},
  booktitle = {{IEEE} International Symposium on Information Theory, {ISIT} 2019,
               Paris, France, July 7-12, 2019},
  pages     = {2464--2468},
  publisher = {{IEEE}},
  year      = {2019},
  url       = {https://doi.org/10.1109/ISIT.2019.8849855},
  doi       = {10.1109/ISIT.2019.8849855},
  bibsource = {dblp computer science bibliography, https://dblp.org}
}

@inproceedings{Babai16,
  author    = {L. Babai},
  editor    = {D. Wichs and
               Y. Mansour},
  title     = {Graph isomorphism in quasipolynomial time [extended abstract]},
  booktitle = {Proceedings of the 48th Annual {ACM} {SIGACT} Symposium on Theory
               of Computing, {STOC} 2016, Cambridge, MA, USA, June 18-21, 2016},
  pages     = {684--697},
  publisher = {{ACM}},
  year      = {2016},
  url       = {https://doi.org/10.1145/2897518.2897542},
  doi       = {10.1145/2897518.2897542},
  bibsource = {dblp computer science bibliography, https://dblp.org}
}

@phdthesis{Saheed-Thesis,
  author  = "M. Saeed{-}Taha",
  title   = "Approche Alge\'{e}brique sur l'\'{e}quivalence des codes",
  school  = "Universit\'{e} de {R}ouen {N}ormandie",
  year    = "2017"
}

@article{PetrankR97,
  author    = {E. Petrank and
               R. Roth},
  title     = {Is code equivalence easy to decide?},
  journal   = {{IEEE} Trans. Inf. Theory},
  volume    = {43},
  number    = {5},
  pages     = {1602--1604},
  year      = {1997},
  url       = {https://doi.org/10.1109/18.623157},
  doi       = {10.1109/18.623157},
  bibsource = {dblp computer science bibliography, https://dblp.org}
}

@misc{Reijnders22,
      author = {K. Reijnders and S. Samardjiska and M. Trimoska},
      title = {Hardness estimates of the Code Equivalence Problem in the Rank Metric},
      howpublished = {Cryptology ePrint Archive, Paper 2022/276},
      year = {2022},
      note = {\url{https://eprint.iacr.org/2022/276}},
      url = {https://eprint.iacr.org/2022/276}
}

@misc{Couvreur21,
  doi = {10.48550/ARXIV.2011.04611},
  url = {https://arxiv.org/abs/2011.04611},
  author = {A. Couvreur and T. Debris-Alazard and P. Gaborit},
  title = {On the hardness of code equivalence problems in rank metric},
  publisher = {arXiv},
  year = {2020},
  copyright = {Creative Commons Attribution 4.0 International}
}

@misc{MEDS,
      author = {T. Chou and R. Niederhagen and E. Persichetti and T. Hajatiana Randrianarisoa and K. Reijnders and S. Samardjiska and M. Trimoska},
      title = {Take your {MEDS}: Digital Signatures from Matrix Code Equivalence},
      howpublished = {Cryptology ePrint Archive, Paper 2022/1559},
      year = {2022},
      note = {\url{https://eprint.iacr.org/2022/1559}},
      url = {https://eprint.iacr.org/2022/1559}
}

@book{KST93,
  author    = {J. K{\"{o}}bler and
               U. Sch{\"{o}}ning and
               J. Tor{\'{a}}n},
  title     = {The Graph Isomorphism Problem: Its Structural Complexity},
  series    = {Progress in Theoretical Computer Science},
  publisher = {Birkh{\"{a}}user/Springer},
  year      = {1993},
  url       = {https://doi.org/10.1007/978-1-4612-0333-9},
  doi       = {10.1007/978-1-4612-0333-9},
  isbn      = {978-1-4612-6712-6},
  bibsource = {dblp computer science bibliography, https://dblp.org}
}

@inproceedings{babai2011code,
  author       = {L{\'{a}}szl{\'{o}} Babai and
                  Paolo Codenotti and
                  Joshua A. Grochow and
                  Youming Qiao},
  editor       = {Dana Randall},
  title        = {Code Equivalence and Group Isomorphism},
  booktitle    = {Proceedings of the Twenty-Second Annual {ACM-SIAM} Symposium on Discrete
                  Algorithms, {SODA} 2011, San Francisco, California, USA, January 23-25,
                  2011},
  pages        = {1395--1408},
  publisher    = {{SIAM}},
  year         = {2011},
  url          = {https://doi.org/10.1137/1.9781611973082.107},
  doi          = {10.1137/1.9781611973082.107},
  bibsource    = {dblp computer science bibliography, https://dblp.org}
}

@misc{ISIT25_CE,
      author = {Huck Bennett and Drisana Bhatia and Jean-François Biasse and Medha Durisheti and Lucas LaBuff and Vincenzo Pallozzi Lavorante and Philip Waitkevich},
      title = {Asymptotic improvements to provable algorithms for the code equivalence problem},
      howpublished = {Cryptology {ePrint} Archive, Paper 2025/187},
      year = {2025},
      url = {https://eprint.iacr.org/2025/187}
}

@misc{BHJ26,
      title={Code Equivalence and Automorphism Problems for Codes}, 
      author={Jean-Francois Biasse and Alexandra V. Hostetler and Anuvrat Jaindungarwal},
      year={2026},
      eprint={2609.25483},
      archivePrefix={arXiv},
      primaryClass={cs.CC},
      url={https://arxiv.org/abs/2609.25483}, 
}
